\documentclass[10pt]{scrartcl} %twoside
\usepackage{authblk}
\usepackage[T1]{fontenc}
\usepackage[utf8]{inputenc}
\usepackage[english]{babel}
\usepackage{xifthen}
\usepackage{amsmath}
\usepackage{header}
\usepackage{comment}
\usepackage{float}
\usepackage{longtable}
\usepackage{thm-restate}
\usepackage{tikz}
\usepackage{wrapfig}
\usepackage{aliases}
\usepackage[left=2.5cm,right=2.5cm,top=3cm,bottom=3cm]{geometry}
\definecolor{myred}{RGB}{236, 17, 0}
\definecolor{myblue}{RGB}{10, 88, 153}
\definecolor{mygreen}{RGB}{26, 152, 81}
\definecolor{myorange}{RGB}{236, 137, 0}
\definecolor{LightGray}{RGB}{220, 220, 220}
\definecolor{LinkColor}{RGB}{167, 20, 49}

\newcommand{\PGM}[4][]{%
    \ifthenelse{\isempty{#1}}%
        {\mathcal{E}^{#2 #3}_{#4}(#3)}%
        {\mathcal{E}^{#1_{#2 #3}}_{#4}(#3)}%
}
\newcommand{\PGMelements}[4][]{%
    \ifthenelse{\isempty{#1}}%
        {\Lambda^{#2 #3;#4}_{#3}}%
        {\Lambda^{#1_{#2 #3};#4}_{#3}}%
}
\newcommand{\PGMchannel}[4][]{%
    \ifthenelse{\isempty{#1}}%
        {\mathcal{E}^{#2 #3}_{#4|#3}}%
        {\mathcal{E}^{#1_{#2 #3}}_{#4|#3}}%
}
\newcommand{\PGMchannelMartin}[2]{
    \Lambda^{#1}_{#2}
}
\newcommand{\TLS}{I}
\newcommand{\textDEOR}{$\mathrm{deor_{\mathcal{A}}}$}

\newcommand{\DEOR}{\mathrm{deor_{\mathcal{A}}}}

\newcommand{\carla}[1]{{\color{blue} Carla: #1}}
\newcommand{\martin}[1]{{\color{red} Martin: #1}}

\title{Improved Two-source Extractors against Quantum Side Information}
\author[1]{Jakob Miller}
\email{jamiller@student.ethz.ch}

\author[1]{Martin Sandfuchs}
\author[1]{Carla Ferradini}
\affil[1]{Institute for Theoretical Physics, ETH Zürich}

\hypersetup{
    colorlinks,
    linkcolor=myblue,
    citecolor=myblue,
    filecolor=myblue,
    urlcolor=myblue,
    pdftitle={Improved Two-source Extractors against Quantum Side Information},
    pdfauthor={Jakob Miller, Martin Sandfuchs, Carla Ferradini}
}

\logo{\includegraphics[width=0.55\columnwidth]{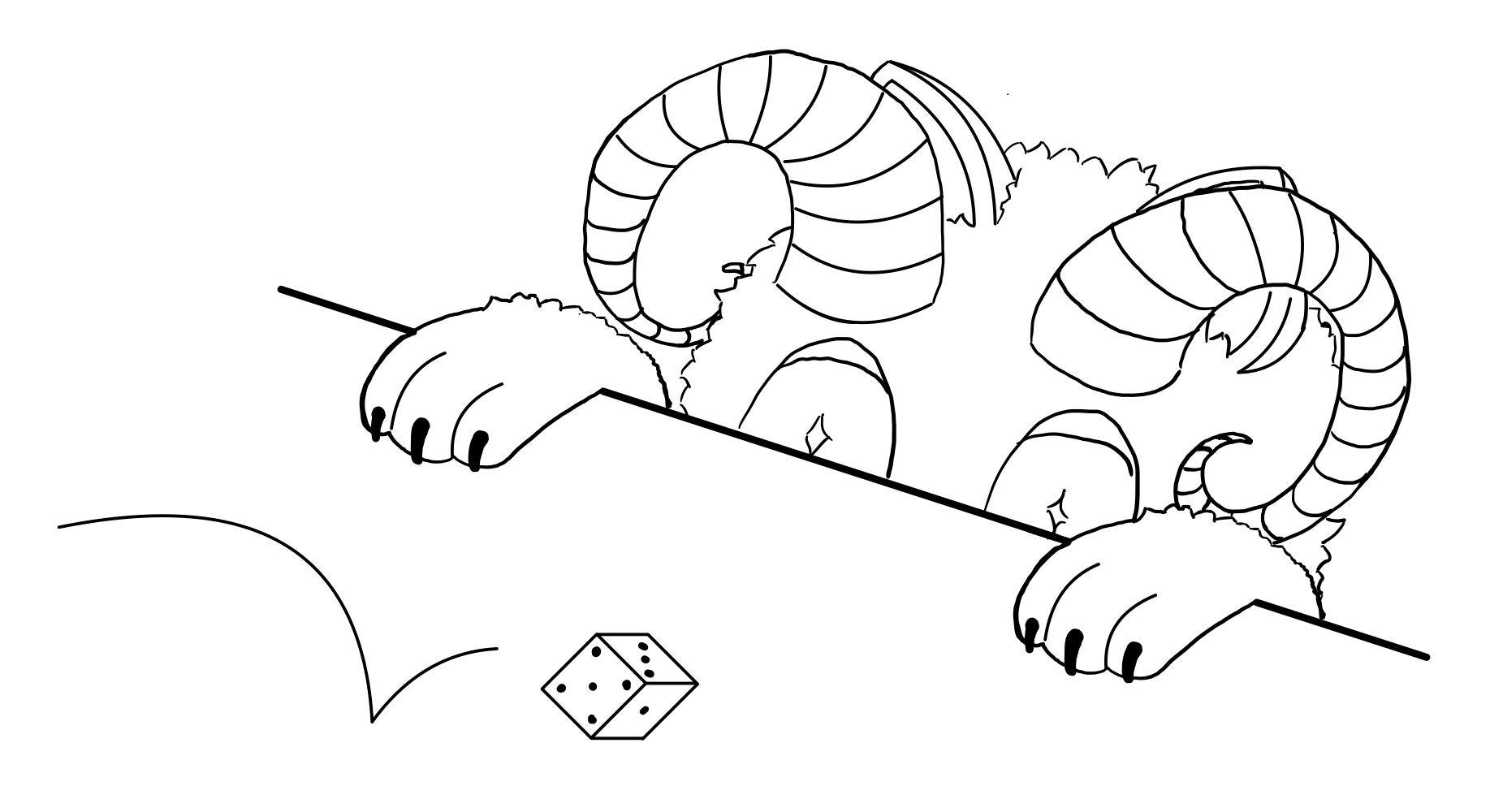}}

\begin{document}

\maketitle
\thispagestyle{empty}

\begin{abstract}
    Two-source extractors aim to extract randomness from two independent sources of weak randomness. It has been shown that any two-source extractor which is secure against classical side information remains secure against quantum side information. Unfortunately, this generic reduction comes with a significant penalty to the performance of the extractor. In this paper, we show that the two-source extractor from \cite{Dodis_extractors} performs equally well against quantum side information as in the classical realm, surpassing previously known results about this extractor. Additionally, we derive a new quantum XOR-Lemma which allows us to re-derive the generic reduction but also allows for improvements for a large class of extractors.
\end{abstract}

% \begin{center}    
% \begin{tikzpicture}[overlay,remember picture]
%     \node[inner sep=0pt] (tmp) at (0,-5) {\includegraphics[width=0.7\textwidth]{figures/Eve-close-up.png}};
% \end{tikzpicture}
% \end{center}

\vspace{2.0cm}
\begin{figure}[h]
    \centering
    \includegraphics[width=0.7\columnwidth]{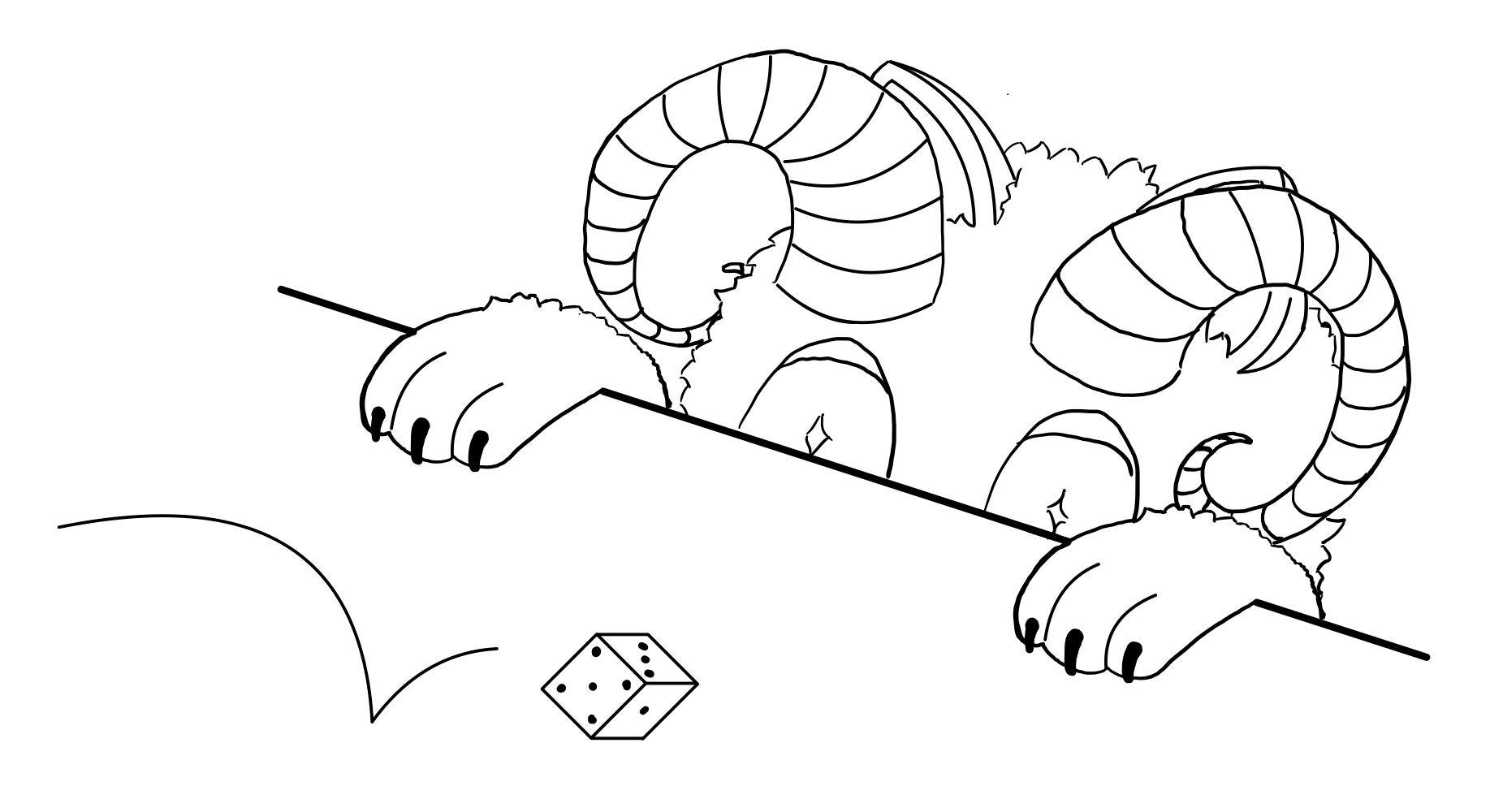}
\end{figure}

{\color{white}{Dear all, latex didn't work properly. Therefore, we had to introduce this phantom text to keep the figure in first page and having a new page for the introduction. If you read this thanks for downloading! Jakob, Martin and Carla}}
\newpage

\section{Introduction}

Randomness extraction is the art of transforming a string of imperfect randomness into a (generally shorter) string of nearly perfect randomness. This allows for constructing sources of near perfect randomness out of weaker resources, making randomness extractors useful tools whenever perfect randomness is required. Furthermore, randomness extractors have found many uses in areas such as derandomization and cryptography. It is well known that it is impossible to extract randomness from a single min-entropy source using a deterministic function. Hence, in the literature one often considers access to an additional short but uniform random seed or access to a second, independent, source of imperfect randomness. These two scenarios are termed seeded extractors and two-source extractors, respectively.

Given that extractors are classical functions acting on classical random variables, one may expect classical information theory to be sufficient for their study. Unfortunately, it has been shown that there are extractors which are secure against classical but not against quantum side information \cite{Gavinsky_2007}. Hence, one has to separately prove security of extractors against the various types of side information. Indeed, there has been a long line of work, showing security of both seeded and two-source extractors against different types of side information. Among the first to consider the security of seeded extractors against quantum adversaries were König, Renner, and Maurer \cite{Renner_PhD_thesis,Konig_2005,RenKoe05}. Later, König and Terhal showed that any 1-bit seeded extractor remains secure against quantum side information with slightly worse parameters \cite{KT08}. For the multi bit output scenario, it was shown in \cite{Berta_2017, Berta_2016} that any seeded extractor with error $\varepsilon$ remains quantum-proof with rescaled error $\sqrt{2^m \varepsilon}$, where $m$ are the number of output bits.

Kasher and Kempe were the first to study security of the more general two-source extractors against quantum adversaries \cite{Kasher_Kempe_main}. They considered adversaries with limited entanglement as well as adversaries with quantum product-type side information. Their main result is that the specific two-source extractor from \cite{Dodis_extractors} remains secure against product-type quantum side information with somewhat worse parameters. Using the XOR-Lemma from \cite{Kasher_Kempe_main}, it was shown in \cite{Chung_2014} that any two-source extractor remains secure with an error scaling as $2^m \sqrt{\varepsilon}$. This scaling was later improved to $\sqrt{2^m \varepsilon}$ using an independent proof technique \cite{arnonfriedman_extractors}. Furthermore \cite{arnonfriedman_extractors} generalized the notion of product side information to the more general Markov model and showed that multi-source extractors remain secure in this more general security model. Two-source extractors in the Markov model have since found applications in device-independent randomness amplification \cite{Kessler_2020,Foreman_2023}.

\subsection{Overview of results}
\begin{table}[h]
    \centering
    \def\arraystretch{1.3}
    \begin{tabular}{| l | l | l |}
        \hline
        Side information & Output length $m$ & Proof/source\\
        \hline
        \hline
        None & $k_1 + k_2 + 2 - n - r - 2\log\frac{1}{\varepsilon}$ & \cite{Dodis_extractors} (\Cref{lem:BleA_ext_no_side_info}) \\\hline
        Classical Markov model & $k_1+k_2+2-n-r-4\log\frac{3}{\varepsilon}$ &  \Cref{cor:BleA_against_classical_knowledge} \\\hline
        Quantum product-type & $\frac{1}{6}(k_1 + k_2 +2 -n - 6\log\frac{1}{\varepsilon})$ & Lemma 30 in \cite{Kasher_Kempe_main} \\\hline
        Quantum product-type & $\frac{1}{2} (k_1 + k_2+3-n-r-4\log\frac{1}{\varepsilon})$ & \Cref{lem:BleA_2m_result} \\\hline
        Quantum product-type & $k_1 + k_2+2-n-r-2\log\frac{1}{\varepsilon}$ & \Cref{thm:BleA_1m_result} \\\hline
        Quantum Markov model & $\frac{1}{5}(k_1 + k_2+9-n-4\log3-8\log\frac{1}{\varepsilon})$ & ~\cite{arnonfriedman_extractors} (\Cref{prop:ext_to_ext_against_quantum_Markov})\\\hline
        Quantum Markov model & $k_1 + k_2+2-n-r-4\log\frac{3}{\varepsilon}$ & \Cref{cor:BleA_against_Markov} \\\hline
    \end{tabular}
    \caption{Output length $m$ for the two-source extractor from \cite{Dodis_extractors} using different proof techniques.
            Here $k_1$ and $k_2$ denote the (conditional) min-entropy of the input sources of imperfect randomness, $X_1$ and $X_2$,  and $n$ denotes the length of the input strings. The parameter $r$ depends on the specific extractor in the \textDEOR{}-extractor family (see~\cref{def:BleA_ext}).
            Note, that in \cite{Kasher_Kempe_main} and \cite{arnonfriedman_extractors} only the case $r=0$ is treated.}
    \label{tab:BleA_results}
\end{table}
Our contribution is two-fold: For ease of reference we will differentiate between two approaches that we call \emph{modular} and \emph{non-modular} in the following.
Firstly, we present a new, \emph{modular} proof-technique that recovers the following known statement.
\begin{theorem}[informal]\label{thm:informal:general_ext}
    Any strong two-source extractor $\Ext$ remains a strong two-source extractor against product-type quantum knowledge with weaker parameters. 
\end{theorem}
The formal statement is given in~\cref{thm:general_ext_result}. Our proof of this statement follows similar steps as \cite[lemma 30]{Kasher_Kempe_main} and \cite[theorem 5.3]{Chung_2014} which can be summarized as follows: Use the XOR-Lemma from \cite{Kasher_Kempe_main} to reduce the security of the $m$-bit extractor $\mathrm{Ext}$ to the case of the single-bit extractor $s \cdot \mathrm{Ext}$. The single-bit case can then be reduced to classical side information using the methods from \cite{KT08}. However, in our proof we use a modified version of the XOR-Lemma which we call the \emph{measured XOR-Lemma}. Just like the XOR-Lemma in \cite{Kasher_Kempe_main}, our generalization relates $\mathrm{Ext}$ to $s \cdot \mathrm{Ext}$ but also applies a pretty good measurement (see~\cref{def:PGM}) to the quantum side information. Since the measured XOR-Lemma already includes a measurement, we can directly apply the classical result, i.e., the measured XOR-Lemma already contains the reduction to classical side information from \cite{KT08}. This then allows us to obtain the improved error of $\sqrt{2^m \varepsilon}$ as in \cite{arnonfriedman_extractors} instead of the looser error $2^m \sqrt{\varepsilon}$ from \cite{Chung_2014}.
Apart from giving a new perspective onto the generic security reduction, our proof-technique also allows for a significant improvement of the possible output length $m$, if $s \cdot \Ext$ is known to be a good extractor against classical product-type side information. We demonstrate this on the example of the extractor from \cite{Dodis_extractors}.
We also note that, in the case of strong extractors, we are able to derive better bounds than \cite{arnonfriedman_extractors}.
The reason for this is an improvement to a lemma reducing security against classical product-type side information to security against no side information (\cref{prop:ext_to_ext_against_classical_product_type}).
Secondly, we show that the specific extractor from \cite{Dodis_extractors}, which we denote by \textDEOR, has identical performance against quantum side information as in the case of no side information:
\begin{theorem}[informal]\label{thm:informal:BleA_1m}
    The \textDEOR{}-extractor is as good in the case of product-type quantum side information as in the case of no side information.
\end{theorem}
For a formal statement we refer to \cref{thm:BleA_1m_result}. 
Furthermore, we demonstrate that this result can be extended to quantum side information in the Markov model, with an insignificant decrease in the extractor's output length, as shown in~\cref{cor:BleA_against_Markov}.
This shows that \textDEOR{} is as good against side information in the quantum Markov model as it is against side information in the classical Markov model. 
Compared to \cite[corollary 21]{arnonfriedman_extractors}, this result constitutes an increase in the output length $m$ by a factor of 5.
In~\cref{tab:BleA_results} we summarize our results for the \textDEOR{}-extractor.

\subsection{Organization of the paper}
The rest of this paper is structured as follows. 
\Cref{chapter:preliminaries} contains the preliminaries clarifying the notation and definitions that we use.
In \cref{sec:two_source_extractors}, we define the security model for two-source extractors and describe how the \textDEOR{}-extractor is constructed.
In~\cref{chapter:results} we present our results, providing formal statements and proofs of~\cref{thm:informal:general_ext,thm:informal:BleA_1m} and of the measured XOR-Lemma.
\Cref{chapter:conclusion} concludes our work and outlines open questions which would make for interesting topics in future research.

\section{Preliminaries}
This section includes definitions, notation and theorems which are used in proving our results. 

\paragraph{Notation.}
\label{chapter:preliminaries}
\label{sec:QM_notation}
\Cref{tab:notation} summarizes some basic notions of quantum information theory as well as the notation used throughout the remainder of this work. For a detailed introduction to the formalism of quantum information theory, see, for instance, \cite{Nielsen_Chuang}. 
    \renewcommand{\arraystretch}{1.5}
\begin{longtable}[H]{|p{0.35\linewidth-5\tabcolsep}||p{0.65\linewidth}|}
        \hline
        \textbf{Notation} & \textbf{Description} \\
        \hline
        \hline
        $\log$ & Logarithm in base 2\\
        \hline 
        $x \cdot y$ & The inner product modulo 2 between $x \in \{0,1\}^n$ and $y \in \{0,1\}^n$ \\
        \hline
        $\mathop{\mathbb{E}}_{x\leftarrow X}[f(x)]$ & Expectation value of a function $f(X)$ of a random variable $X$ with distribution $P_X(x)$ \\
        \hline

        $\hilmap_A, \hilmap_X$, ... & Hilbert spaces belonging to different systems, $A$, $X$, ... (typically $A$,$B$, $\dots$ are quantum systems and $X$, $Y$, $\dots$ are classical)\footnote{In quantum information theory, systems are associated to Hilbert spaces. Hence, we identify systems by the corresponding system label, i.e., whenever we say ``system A'' we refer to the quantum system represented by the Hilbert space $\mathcal{H}_{A}$. The state of a system is described by density operators. The evolution of quantum system is described by completely positive and trace-preserving (CPTP) maps.} \\
        \hline
        $\mathrm{End}(\hilmap)$ & Set of linear operators on $\hilmap$ \\
        \hline
        $\mathcal{P}(\hilmap)$ & Set of positive semi-definite linear operators on $\hilmap$ \\
        \hline
        $\mathcal{S}(\hilmap)$ & Set of positive semi-definite, trace-one, linear operators on $\hilmap$, i.e., density operators\\
        \hline
        $\rho_X=\sum_{x\in\mathcal{X}}P_X(x)\ketbra{x}{x}_X$ & Classical state in the preferred fixed basis $\{\ket{x}\}_{x\in\mathcal{X}}$\\
        \hline
        $\rho_{XB}$ \hfill 
       
        $=\sum_{x\in\mathcal{X}} P_X{(x)} \ketbra{x}{x}_X \otimes \rho_{B|x}$ & Classical-quantum state (cq-state) in $\mathcal{S}(\hilmap_X\otimes\hilmap_B)$ with conditional states $\rho_{B|x}\in\mathcal{S}(\hilmap_B)$ and probability distribution $P_X$\\
        \hline
        $\rho_{B \land x} = P_X(x) \rho_{B|x}$ & Subnormalized conditional state relative to the cq state $\rho_{XB}$\\
        \hline
        $\omega_{d} \coloneq \frac{\id}{d}$ & Maximally mixed state on $\hilmap$ where $d=\dim(\hilmap)$ \\
        \hline
        $[v] \coloneq \ketbra{v}{v}$ & Rank one projector onto the vector $\ket{v}$\\
        \hline
        $\mathrm{CPTP}(\hilmap_A, \hilmap_B)$&
        Set of completely-positive and trace-preserving linear maps from $\mathrm{End}(\hilmap_A)$ to $\mathrm{End}(\hilmap_B)$\\
        \hline
        $\mathcal{E}_{B|A}$ & An element of the set $\mathrm{CPTP}(\hilmap_A, \hilmap_B)$\footnote{For a bipartite state $\rho_{AC} \in \mathcal{S}(\hilmap_A\otimes\hilmap_C)$ and a channel $\mathcal{E}_{B|A}\in\mathrm{CPTP}(\hilmap_A, \hilmap_B)$ that acts only on subsystem $A$, we will typically omit the identity channel acting on subsystem $C$, i.e., we write $\mathcal{E}_{B|A}(\rho_{AC})$ instead of $(\mathcal{E}_{B|A}\otimes\mathcal{I}_{C})(\rho_{AC})$.} \\
        \hline
        $\mathcal{I}_{A}$ & Identity channel on linear operators on $\hilmap_{A}$, i.e., for all $\rho_A\in\mathcal{S}(\hilmap_A)$, $\mathcal{I}_{A}(\rho_A)=\rho_A$\\
        \hline
        $f_{Y|X}[\rho_{X}] = f[\rho_{X}] = \rho_{f(X)}$ \hfill
        
        $ = \sum_{x} \ketbra{f(x)}{f(x)} \bra{x} \rho \ket{x}$ & CPTP map corresponding to a classical function $f: \mathcal{X} \mapsto \mathcal{Y}$ \\
        \hline
        $\norm{S}_1 \coloneq \tr{\abs{S}} \coloneq \tr{\sqrt{S^\dag S}}$& Trace norm $\norm{\cdot}_1$ of $S\in \mathrm{End}(\hilmap)$\\
        \hline
        $\delta(\rho,\sigma) \coloneq \frac{1}{2}\norm{\rho-\sigma}_1$& Trace distance between $\rho, \sigma \in \mathcal{S}(\hilmap)$ (\cref{def: trace distance})\\
        \hline
        $d_2(\rho_{AB}|\sigma_B)$ & $L_2$-distance to uniform (\cref{def: L2 distance}) \\
        \hline
        $H_{\min}(\rho_{AB}|\sigma_B), H_{\min}(A|B)_\rho$ & Conditional min entropy (\cref{def:quantum_cond_min_entropy}) \\
        \hline
        $H_2(\rho_{AB}|\sigma_B),H_2(A|B)\rho$ & Collision entropy (\cref{def:collision entropy}) \\
        \hline
      
    \caption{Summary of notation.}
    \label{tab:notation}
\end{longtable}

\subsection{Distance measures}\label{sec:distance_states}
Quantifying the level of imperfection of a randomness source amounts to quantifying the distance from its state to an ideal state. Therefore, in this section we introduce distance measures between quantum states and their relevant properties. 

\begin{definition}[Trace distance]
    \label{def: trace distance}
    We define the trace distance between two density operators $\rho, \sigma \in \mathcal{S}(\hilmap)$ as 
    \begin{equation*}
        \delta(\rho,\sigma) = \frac{1}{2}\norm{\rho-\sigma}_1,
    \end{equation*}
    where for any $S\in \mathrm{End}(\hilmap)$ the trace norm $\norm{\cdot}_1$ is given by  
    \begin{equation*}
        \norm{S}_1 \coloneq \tr{\abs{S}} \coloneq \tr{\sqrt{S^\dag S}}.
    \end{equation*}  
\end{definition}

The following statement is immediately obtained from the definition of the trace distance.
\begin{proposition}\label{prop:trace_norm_cq}
    Let $\rho_{XB},\sigma_{XB}\in\mathcal{S}(\hilmap_X\otimes\hilmap_B)$ be cq-states such that
    \begin{equation*}
        \begin{split}
            \rho_{XB} = \sum_{x\in\mathcal{X}} P_X{(x)} \ketbra{x}{x}_X \otimes \rho_{B|x}\quad \text{ and }\quad
            \sigma_{XB} = \sum_{x\in\mathcal{X}} P_X{(x)} \ketbra{x}{x}_X \otimes \sigma_{B|x},
        \end{split}
    \end{equation*}
    i.e., the states have the same distribution $P_{X}(x)$, then
    \begin{equation*}
        \delta\bigl(\rho_{XB} , \sigma_{XB} \bigr) = \mathop{\mathbb{E}}_{x\leftarrow X} \mleft[ \delta\bigl(\rho_{B|x} , \sigma_{B|x} \bigr) \mright].
    \end{equation*}
\end{proposition}
A proof of the following lemma can be found in~\cite[section~9.2]{Nielsen_Chuang}.
\begin{lemma}[Data-processing]\label{lem:trace_norm_data_processing}
    For any channel $\mathcal{E}_{B|A}\in\mathrm{CPTP}(\hilmap_A,\hilmap_B)$ and any hermitian operator $S_A\in\mathrm{End}(\hilmap_A)$, we have $\norm{\mathcal{E}_{B|A}[S_A]}_1 \leq \norm{S_A}_1$.
\end{lemma}

In what follows, we will use the trace distance to quantify the imperfection of a randomness source. In order to derive upper-bounds on the trace distance, we use the $L_2$-distance to uniform originally defined in~\cite{Renner_PhD_thesis}.

\begin{definition}[$L_2$-distance to uniform]
\label{def: L2 distance}
    Let $\rho_{AB}\in \mathcal{P}(\hilmap_A\otimes \hilmap_B)$ and $\sigma_B\in\mathcal{P}(\hilmap_{B})$. We define the (conditional) $L_2$-distance to uniform of $\rho_{AB}$ relative to $\sigma_B$ as
    \begin{equation*}
        d_2(\rho_{AB}|\sigma_B) \coloneq \tr{\mleft((\id_A \otimes \sigma^{-1/4}_B)(\rho_{AB} - \omega_{A} \otimes \rho_B)(\id_A \otimes \sigma^{-1/4}_B)\mright)^2},
    \end{equation*}
    where $\omega_{A}$ denotes the maximally mixed state on $A$.
\end{definition}
In the case of cq-states, the conditional $L_2$-distance can be expressed more compactly. 
\begin{lemma}[{\cite[Lemma 5.2.4]{Renner_PhD_thesis}}]\label{lem:d_2_cq_states}
    Let $\rho_{XB}\in\mathcal{S}(\hilmap_X \otimes \hilmap_B)$ be a cq-state with conditional operators $\rho_{B \land x}$, i.e., $\rho_{XB} =\sum_{x\in\mathcal{X}} [x]_X\otimes\rho_{B \land x}$.
    Let $\sigma_B\in\mathcal{P}(\hilmap_B)$, then
    \begin{equation*}
        d_2(\rho_{XB}|\sigma_B) = \sum_{x\in\mathcal{X}}\tr{\mleft( \sigma_B^{-1/4} \rho_{B \land x} \sigma_B^{-1/4} \mright)^2} - \frac{1}{|\mathcal{X}|}\tr{\mleft( \sigma_B^{-1/4} \rho_B \sigma_B^{-1/4} \mright)^2}.
    \end{equation*}
\end{lemma}
The following statement allows us to derive upper bounds on the trace distance using the $L_2$ distance:

\begin{corollary}[{\cite[Lemma 5.2.3]{Renner_PhD_thesis}}]\label{cor:bound_one_norm_by_two_norm}
    Let $\rho_{AB} \in \mathcal{P}(\hilmap_A \otimes \hilmap_B)$ and $\sigma_B \in \mathcal{P}(\mathcal{H_B})$. Then\footnote{
        In principle, one has to be careful when dealing with non-invertible $\sigma_{B}$. However, it is readily verified that the statement holds if $\mathrm{ker}(\id_{A} \otimes \sigma_{B}) \subseteq \mathrm{ker}(\rho_{AB})$ which will always be true in our applications.}
    \begin{equation*}
        \delta\big(\rho_{AB},\omega_{A}\otimes \rho_B\big) \leq \frac{1}{2}\sqrt{\mathrm{dim}(\hilmap_A)\tr{\sigma_B}d_2(\rho_{AB}|\sigma_B) },
    \end{equation*}
    where $\omega_{A}$ denotes the maximally mixed state on $A$.
\end{corollary}

\subsection{Entropies}\label{sec:entropies}
\begin{definition}[Conditional min-entropy] \label{def:quantum_cond_min_entropy}
    Let $\rho_{AB}\in\mathcal{P}(\hilmap_A \otimes \hilmap_B)$ and $\sigma_B\in\mathcal{P}(\hilmap_B)$. If $\ker(\id_{A} \otimes \sigma_{B}) \subseteq \rho_{AB}$, then the min-entropy of $\rho_{AB}$ relative to $\sigma_B$ is defined as 
    \begin{align*}
        \begin{split}
            H_{\min}(\rho_{AB}|\sigma_B) \coloneq& -\log\min\{\lambda : \rho_{AB} \leq \lambda \id_A \otimes \sigma_{B}\} = \max \{ z : \rho_{AB} \leq 2^{-z} \id_{A} \otimes \sigma_{B} \},
        \end{split}
    \end{align*}
    and as $-\infty$ if $\ker(\id_{A} \otimes \sigma_{B}) \not\subseteq \rho_{AB}$.
    The min-entropy of $\rho_{AB}$ given $B$ is defined as 
    \begin{equation*}
        H_{\min}(A|B)_\rho \coloneq \max_{\sigma_B} H_{\min}(\rho_{AB}|\sigma_B),
    \end{equation*}
    where the maximum is taken over all $\sigma_B\in\mathcal{S}(\hilmap_B)$.
\end{definition}

\begin{lemma}[Data-processing]\label{lem:data_processing}
    For any channel $\mathcal{E}_{B'|B}\in\mathrm{CPTP}(\hilmap_B,\hilmap_{B'})$ and for any state $\rho_{AB} \in \mathcal{S}(\hilmap_A\otimes\hilmap_B)$, we have
    \begin{equation*}
        H_{\min}(A|B)_{\rho} \leq H_{\min}(A|B')_{\sigma}
    \end{equation*}
    with $\sigma_{AB'} = \mathcal{E}_{B'|B}(\rho_{AB})$.
\end{lemma}

If, instead of manipulating the conditioning system $B$, we apply a function to system $A$, we would expect our knowledge about $A$ to increase.
Consequently, the entropy should decrease.
Indeed, for the specific case of two classical systems and a linear transformation acting on the first system we get the following result which is a generalization of \cite[lemma 3]{Dodis_extractors}.
\begin{proposition}\label{prop:min_entropy_matrix_multipl}

    Let $\rho_{XB}$ be a cq-state where $X$ takes values in $\{0, 1\}^{n}$. Consider a $n \times n$ matrix $L$ with entries in $\{0, 1\}$ and $\mathrm{rank}(L) \geq n - r$. Then
    \begin{equation*}
        H_{\min}(X'|B)_{\rho_{(L \cdot X) B}} \geq H_{\min}(X|B)_{\rho_{XB}} - r,
    \end{equation*}
    where $X'= L\cdot X$.
\end{proposition}
\begin{proof}

Let $z = H_{\min}(X|B)_{\rho}$ and $\sigma_B \in \mathcal{S}(\hilmap_B)$ be such that $\rho_{XB} \leq 2^{-z} \id_{X} \otimes \sigma_B$ (i.e., $\sigma_B$ is the optimizer for the min-entropy). In particular, this implies that $\rho_{B \land x} \leq 2^{-z} \sigma_{B}$ holds for all $x$. We have that
\begin{equation} \label{eq:min_ent_mat_mult_quant}
\begin{aligned}
    \rho_{(L \cdot X) B} =& \sum_{x'} [x']_{X'} \otimes \Bigl(\sum_{\substack{x\in\{0,1\}^n \\\mathrm{s.t.}\; Lx=x'}}\rho_{B \land x} \Bigr) \leq \sum_{x'} [x']_{X'} \otimes \Bigl( \sum_{\substack{x\in\{0,1\}^n \\\mathrm{s.t.}\; Lx=x'}} 2^{-z} \sigma_{B} \Bigr) \leq 2^{-z}2^{r} \id_{X'} \otimes \sigma_{B},
\end{aligned}
\end{equation}
where for the first inequality we used that $\rho_{B \land x} \leq 2^{-z} \sigma_B$ and for second inequality we noticed that if $\mathrm{rank}(L) \geq n-r$, then $L$ can map at most $2^r$ different inputs $x$ to the same value $x'$. The desired inequality then follows immediately from~\cref{eq:min_ent_mat_mult_quant} and the definition of the min-entropy (\cref{def:quantum_cond_min_entropy}).
\end{proof}

\begin{definition}[Collision-entropy]
\label{def:collision entropy}
    Let $\rho_{AB}\in\mathcal{P}(\hilmap_A \otimes \hilmap_B)$ and $\sigma_B\in\mathcal{P}(\hilmap_B)$. If $\ker(\id_{A} \otimes \sigma_B) \subseteq \ker(\rho_{AB})$, then the collision-entropy of $\rho_{AB}$ relative to $\sigma_B$ is defined as 
    \begin{equation*}
        H_2(\rho_{AB}|\sigma_B) \coloneq -\log \frac{1}{\tr{\rho_{AB}}} \tr{\mleft((\id_{A} \otimes\sigma_B^{-1/4})\rho_{AB}(\id_{A} \otimes\sigma_B^{-1/4})\mright)^2}
    \end{equation*}
    and as $-\infty$ if $\ker(\id_{A} \otimes \sigma_{B}) \not\subseteq \ker(\rho_{AB})$.
    The collision-entropy of $\rho_{AB}$ given $B$ is defined as
    \begin{equation*}
        H_2(A|B)_\rho = \sup_{\sigma_B} H_2(\rho_{AB}|\sigma_B)
    \end{equation*}
    where the maximum is taken over all $\sigma_B\in\mathcal{S}(\hilmap_B)$.
\end{definition}
As shown in \cite[remark 5.3.2]{Renner_PhD_thesis}, the collision-entropy provides an upper bound to the min-entropy, i.e.,
\begin{equation*}
    H_{\min}(\rho_{AB}|\sigma_B)\leq H_2(\rho_{AB}|\sigma_B),
\end{equation*}
which immediately implies
\begin{equation*}
    H_{\min}(A|B)_\rho\leq H_2(A|B)_\rho.
\end{equation*}
In the scenario where system $A$ is classical, the collision-entropy simplifies to the expression given below, which was shown in \cite[remark 5.3.3]{Renner_PhD_thesis}.
\begin{lemma}
\label{lem:collision_ent_cq}
    Let $\rho_{XB}\in\mathcal{P}(\hilmap_X\otimes\hilmap_B)$ be a (non-normalized) cq-state that is classical on $X$.
    Let $\sigma_B\in\mathcal{P}(\hilmap_B)$, then
    \begin{equation*}
        H_{2}(\rho_{XB}|\sigma_B) = - \log \frac{1}{\tr{\rho_{XB}}} \sum_{x} \tr{\mleft(\sigma_B^{-1/4} \rho_{B \land x}\sigma_B^{-1/4}\mright)^2}.
    \end{equation*}
\end{lemma}

\section{Security model and the \textDEOR{}-extractor} \label{sec:two_source_extractors}
As discussed in the introduction, there are different security models for two-source extractors considered in the literature. 
In this section, we present two-source extractors without side information, against classical/quantum product-type side information and against classical/quantum side information in the Markov model.

\subsection{No side information}
Let us begin with the case of two sources of randomness, that are imperfect, and where there is no adversary collecting side information. 
\begin{definition}[Two-source extractor] \label{def:ext}
    A function $\Ext : \{0,1\}^{n_1} \times \{0,1\}^{n_2} \mapsto \{0,1\}^{m}$ is said to be a $(k_1,k_2,\varepsilon)$ weak two-source extractor if for all independent sources $X_1$ and $X_2$ that satisfy $H_{\min}(X_1)\geq k_1$ and  $H_{\min}(X_2)\geq k_2$,
    \begin{equation*}
        \delta\big(\rho_{\Ext(X_1,X_2)},\omega_{m}\big)\leq \varepsilon
    \end{equation*}
    is fulfilled. 
    %We denote $\rho_{\Ext(X_1,X_2)} = \Ext_{X_1X_2}\rho_{X_1 X_2}$. 
    The extractor is called $X_i$-strong for $i=1,2$ if
    \begin{equation*}
        \delta\big(\rho_{\Ext(X_1,X_2)X_i}, \omega_{m} \otimes \rho_{X_i}\big)\leq \varepsilon,
    \end{equation*}
    %where $\rho_{\Ext(X_1,X_2)X_i} = (\Ext_{X_1X_2}\otimes\mathcal{I}_{X_i})\rho_{X_1 X_2 X_i}$ 
    where the second $X_i$ is an additional register containing a copy of the first $X_i$ register.
\end{definition}
% In the definition above, the function $\Ext$ is understood as a CPTP map as described in~\cref{sec:QM_notation}.
% Independence of the sources $X_1$ and $X_2$ implies that $P_{X_1 X_2}(x_1,x_2)=P_{X_1}(x_1)P_{X_2}(x_2)$ and
% \begin{equation*}
%     \rho_{X_1 X_2} = \sum_{x_1 \in \mathcal{X}_1, x_2\in\mathcal{X}_2} P_{X_1}(x_1)P_{X_2}(x_2)\mleft([x_1]_{X_1}\otimes[x_2]_{X_2}\mright).
% \end{equation*}  %For compactness, we will typically write $[x_1]$ instead of $\ketbra{x_1}{x_1}_{X_1}$ from here on. 

\subsection{Product-type side information}
If we allow the adversary to obtain some classical or quantum side information about the randomness sources, but demand independence between the side information on source $X_1$ and the side information on source $X_2$, we get the following definition.
\begin{definition}[Two-source extractor against quantum (classical) product-type knowledge] \label{def:ext_qc_know}
    A function $\Ext : \{0,1\}^{n_1} \times \{0,1\}^{n_2} \to \{0,1\}^{m}$ is said to be a $(k_1,k_2,\varepsilon)$ weak two-source extractor against quantum (classical) product-type knowledge if for all sources $X_1$, $X_2$ and quantum (classical) side information $C=C_1 C_2$ of the form $\rho_{X_1 X_2 C}=\rho_{X_1 C_1}\otimes\rho_{X_2 C_2}$ that satisfies $H_{\min}(X_1|C_1)\geq k_1$ and  $H_{\min}(X_2|C_2)\geq k_2$,
    \begin{equation*}
        \delta\big(\rho_{\Ext(X_1,X_2) C},\omega_{m}\otimes\rho_C\big)\leq \varepsilon
    \end{equation*}
    is fulfilled. 
    %We denote $\rho_{\Ext(X_1,X_2) C} = (\Ext_{X_1X_2}\otimes\mathcal{I}_C)\rho_{X_1 X_2 C}$.
    The extractor is called $X_i$-strong for $i=1,2$ if
    \begin{equation*}
        \delta\big(\rho_{\Ext(X_1,X_2)X_i C}, \omega_{m}\otimes\rho_{X_i C}\big)\leq \varepsilon,
    \end{equation*}
    %where $\rho_{\Ext(X_1,X_2) X_i C} = (\Ext_{X_1X_2}\otimes\mathcal{I}_{X_i C})\rho_{X_1 X_2 X_i C}$ 
    where the second $X_i$ is an additional register containing a copy of the first $X_i$ register.
\end{definition}
% For product-type knowledge we have $\rho_{X_i C}=\rho_{X_i C_i}\otimes\rho_{C_{\bar{i}}}$, where the index $\bar{i}$ denotes the remaining index that is different from $i$.
Note, that using \cref{prop:trace_norm_cq} we can rewrite the security condition for an $X_1$-strong extractor as
\begin{equation}\label{eq:strong_ext_condition}
    \begin{split}
        \delta\big(\rho_{\Ext(X_1,X_2)X_1 C_1 C_2}, \omega_{m}\otimes\rho_{X_1 C_1}\otimes\rho_{C_2}\big)
        &= \mathop{\mathbb{E}}_{x_1\leftarrow X_1}\mleft[ \delta\bigl(\rho_{\Ext(x_1,X_2) C_2}\otimes\rho_{C_1|x_1}, \omega_{m}\otimes\rho_{C_2}\otimes\rho_{C_1|x_1}\bigr) \mright]\\
%           &= \frac{1}{2}\Biggl\lVert\sum_{x_1,x_2} P_{X_1}(x_1)P_{X_2}(x_2) \Big([\Ext(x_1,x_2)]\otimes [x_1] \otimes \rho_{C_1}^{x_1}\otimes \rho_{C_2}^{x_2} - \rho_{U_m}\otimes [x_1] \otimes\rho_{C_1}^{x_1}\otimes\rho_{C_2}^{x_2}\Big)\Biggr\rVert _1\\
%           &= \mathop{\mathbb{E}}_{x_1\leftarrow X_1}\mleft[ \frac{1}{2}\norm{\sum_{x_2} P_{X_2}(x_2)\mleft([\Ext(x_1,x_2)] \otimes \rho_{C_1}^{x_1} \otimes \rho_{C_2}^{x_2} - \rho_{U_m}\otimes\rho_{C_1}^{x_1}\otimes\rho_{C_2}^{x_2}\mright)}_1 \mright]\\
%           &= \sum_{x_1} P_{X_1}(x_1)\frac{1}{2}\norm{\sum_{x_2} P_{X_2}(x_2)\mleft([\Ext(x_1,x_2)] \otimes \rho_{C_2}^{x_2} - \rho_{U_m}\otimes\rho_{C_2}^{x_2}\mright)}_1\\
        &= \mathop{\mathbb{E}}_{x_1\leftarrow X_1} \mleft[\delta\big(\rho_{\Ext(x_1,X_2)C_2}, \omega_{m}\otimes\rho_{C_2}\big)\mright] \\
        &= \delta \big( \rho_{\mathrm{Ext}(X_1, X_2) X_1 C_2} , \omega_{m} \otimes \rho_{X_1 C_2} \big),
    \end{split}
\end{equation}
where $\rho_{\Ext(x_1, X_2)C_2}$ is the state obtained by interpreting the function $\Ext(x_1,\cdot)$ for fixed $x_1$ as a channel acting on $X_{2}$.
Following the same argument as in \cite[proposition 1]{KT08}, we can show that a two-source extractor always yields a two-source extractor against classical product-type knowledge with weaker parameters. %\carla{Remark on conditional entropy}
For convenience, we reproduce the proof in \cref{app:proof:ext_to_ext_against_classical_product_type}.
\begin{restatable}{proposition}{extagainstcpt}
    \label{prop:ext_to_ext_against_classical_product_type}
    A $(k_1,k_2,\varepsilon)$ $X_1$-strong two-source extractor is a $(k_1, k_2 + \log\frac{1}{\varepsilon}, 2\varepsilon)$ $X_1$-strong two-source extractor against classical product-type knowledge.
\end{restatable}

An analogous result can be proven for $X_2$-strong extractors.

\subsection{Side information in the Markov model}
Arnon-Friedman et al.\ \cite{arnonfriedman_extractors} introduced the Markov model as a generalization of product-type side information.
Instead of full independence between the side information on source $X_1$ and the side information on $X_2$, they demand that, conditioned on the side information, the two sources have to be independent.
This condition is formalized by the notion of a Markov chain that we give below. 
\begin{definition}
    A state $\rho_{A_1 A_2 C}\in\mathcal{S}(\hilmap_{A_1}\otimes\hilmap_{A_2}\otimes\hilmap_{C})$ is called a Markov chain, if 
    \begin{equation*}
        I(A_1:A_2|C)=0.
    \end{equation*}
    When referring to such a Markov chain we write $A_1\leftrightarrow C\leftrightarrow A_2$.
\end{definition}
\begin{theorem}[{\cite[Theorem 6]{Markov_rewriting}}]
\label{thm:Markov_rewriting}
    Let $\rho_{A_1 A_2 C}\in\mathcal{S}(\hilmap_{A_1}\otimes\hilmap_{A_2}\otimes\hilmap_{C})$ be such that $I(A_1:A_2|C)=0$, i.e., $A_1\leftrightarrow C\leftrightarrow A_2$ is a Markov chain. Then, there exists a decomposition of $\hilmap_{C}$ as
    \begin{equation*}
        \hilmap_{C} = \bigoplus_{z\in\mathcal{Z}}\hilmap_{C_1^z}\otimes \hilmap_{C_2^z},
    \end{equation*}
    such that
    \begin{equation*}
        \rho_{A_1 A_2 C} = \mathop{\bigoplus}_{z\in\mathcal{Z}}P_Z(z) \rho^z_{A_1 C_1^z} \otimes \rho^z_{A_2 C_2^z},
    \end{equation*}
    where $\rho^z_{A_1 C_1^z}\in\mathcal{S}(\hilmap_{A_1}\otimes\hilmap_{C_1^z})$, $\rho^z_{A_2 C_2^z}\in\mathcal{S}(\hilmap_{A_2}\otimes\hilmap_{C_2^z})$ and $\{P_Z(z)\}_{z\in\mathcal{Z}}$ is a probability distribution.
\end{theorem}
%\carla{I think we never use this and we can remove it: If $\rho_{A_1 A_2 C}$ fulfills the above condition and systems $A_1$, $A_2$ and $C$ are classical, we say that $\rho_{A_1 A_2 C}$ forms a classical Markov chain.}

Since we are interested in randomness sources $X_1$ and $X_2$ and side information $C$ that form a Markov chain, we will, from hereon, only consider Markov chains $\rho_{X_1 X_2 C}$ where $X_1$ and $X_2$ are classical systems, while system $C$ may be quantum or classical.
% Applying the decomposition in~\cref{eq:Markov_rewriting} to such a state, we immediately observe that for a trivial system $Z$ the state $\tilde{\rho}_{X_1 X_2 C_1 C_2 Z}$ reduces to product-type side information.
% This shows that the Markov model can, indeed, be seen as a generalization of product-type knowledge.
With the above definitions, we can introduce the Markov model for two-source extractors.

\begin{definition}
    Consider any two randomness sources $X_1$ and $X_2$ and quantum (classical) side information $C$, such that $H_{\mathrm{min}}(X_1|C)\geq k_1$ and $H_{\mathrm{min}}(X_2|C)\geq k_2$ and $X_1 \leftrightarrow C \leftrightarrow X_2$ form a Markov chain. A function $\Ext: \{0,1\}^{n_1} \times \{0,1\}^{n_2} \mapsto \{0,1\}^m$ is said to be a $(k_1,k_2,\varepsilon)$ weak two-source extractor against quantum (classical) knowledge in the Markov model, if for all such sources, we have
    \begin{equation*}
        \delta\bigl( \rho_{\Ext(X_1,X_2)C}, \omega_{m}\otimes\rho_C\bigr)\leq\varepsilon.
    \end{equation*}
    The extractor is called $X_i$-strong if the following condition holds:
    \begin{equation*}
        \delta\bigl( \rho_{\Ext(X_1,X_2)X_i C}, \omega_{m}\otimes\rho_{X_i C}\bigr)\leq\varepsilon,
    \end{equation*}
     where the second $X_i$ is an additional register containing a copy of the first $X_i$ register.
\end{definition}

With a proof similar to that of~\cref{prop:ext_to_ext_against_classical_product_type} one finds the following proposition.
\begin{proposition}[{\cite[Lemma 6]{arnonfriedman_extractors}}]\label{prop:ext_to_ext_against_classical_Markov}
    A $(k_1,k_2,\varepsilon)$ $X_1$-strong two-source extractor is a $(k_1 + \log\frac{1}{\varepsilon}, k_2 + \log\frac{1}{\varepsilon}, 3\varepsilon)$ $X_1$-strong two-source extractor against classical knowledge in the Markov model.
\end{proposition}
% For convenience, we give a proof in~\cref{app:proof:ext_to_ext_against_classical_Markov}.
% Analogous results can be obtained for a $X_2$-strong or weak extractor.
To our knowledge, the following lemma from~\cite{arnonfriedman_extractors} provides the latest improvement to the error scaling of strong two-source extractors against quantum knowledge in the Markov model.
\begin{proposition}[{\cite[Lemma 9]{arnonfriedman_extractors}}]\label{prop:ext_to_ext_against_quantum_Markov}
    A $(k_1,k_2,\varepsilon)$ $X_1$-strong two-source extractor is a $(k_1 + \log\frac{1}{\varepsilon}, k_2 + \log\frac{1}{\varepsilon},\sqrt{3 \cdot 2^{m-2} \varepsilon})$ $X_1$-strong two-source extractor against quantum knowledge in the Markov model, where $m$ is the output length of the extractor.
\end{proposition}

For applications, one often only has a bound on the smooth min-entropy rather than the min-entropy. It is therefore crucial to construct extractors which can work with smooth min-entropy sources. It was shown in \cite[lemma 17]{arnonfriedman_extractors} that any extractor in the Markov model performs almost identically even if one only has a guarantee on the smooth min-entropy.

%\carla{Maybe reverence lemma 9 or corollary 21 of Rotem's}

\subsection{The \textDEOR{}-extractor family} \label{sec:Dodis_ext}
% For a concrete statement showing that the inner product gives a good two-source extractor we refer to Appendix \ref{app:B} and Corollary \ref{cor:IP_strong_cl_ext}, in particular.

In the remainder of this section we want to define a family of extractors which was originally introduced in \cite{Dodis_extractors}. 
% It can be understood as a generalization of the inner product extractor. 
\begin{definition}\label{def:BleA_ext}
    Let $\mathcal{A}\coloneq \{A_i\}_{i=1}^{m}$ be a set of $n\times n$-matrices with entries in the finite field with two elements\footnote{The elements of such a field are denoted by $\{0,1\}$. Addition and multiplication are defined respectively as the usual addition and multiplication of integers modulo $2$.}, such that, for every $s\in\{0,1\}^m \setminus \{0^m\}$, the matrix $A_{s}\coloneqq \sum_{i=1}^m A_i s_i$ satisfies $\mathrm{rank}(A_s)\geq n-r$. Define
    \begin{align*}
        \DEOR: \{0,1\}^n \times \{0,1\}^n & \longrightarrow \{0,1\}^m \\
        (x, y) & \longmapsto ((A_1 x)\cdot y,...,(A_m x)\cdot y) = (x^T A_1^T y,..., x^T A_m^T y)
    \end{align*}
    with $\cdot$ the inner product modulo 2.
\end{definition}
It can be shown that this definition indeed gives a classical two-source extractor. 
\begin{lemma}[{\cite[Theorem 1]{Dodis_extractors}}]\label{lem:BleA_ext_no_side_info}
    The function \textDEOR{} is a $(k_1,k_2,\varepsilon)$ $X_2$-strong two-source extractor with
    \begin{equation*}
        \varepsilon = 2^{-\frac{k_1 + k_2 + 2 - n - r - m}{2}}.
    \end{equation*}
\end{lemma}
Note, that exchanging the sources $X_1$ and $X_2$ with one another yields almost the same extractor with the only difference being that the matrices in $\mathcal{A}$ need to be transposed. 
Since for any matrix $A$ we have $\mathrm{rank}(A)=\mathrm{rank}(A^T)$, we immediately find that the \textDEOR-extractor is strong in both $X_1$ and $X_2$ separately.
% Following this argument and using the fact that the definition of the Markov model is symmetric under exchanging sources $X_1$ and $X_2$, all of the following results about the \textDEOR-extractor being $X_1$-strong imply that it is $X_2$-strong with the same parameters.
\begin{corollary}\label{cor:BleA_against_classical_knowledge}
    The function \textDEOR{} is a $(k_1,k_2,\varepsilon)$ $X_1$-strong two-source extractor against classical knowledge in the Markov model with
    \begin{equation*}
        \varepsilon = 2^{-\frac{k_1+k_2+2-4\log3-n-r-m}{4}}.
    \end{equation*}
\end{corollary}
\begin{proof}
    From~\cref{prop:ext_to_ext_against_classical_Markov} we know that any $(k_1',k_2',\varepsilon')$ $X_1$-strong two-source extractor is a $(k_1'+\log\frac{1}{\varepsilon},k_2'+\log\frac{1}{\varepsilon'},3\varepsilon')$ $X_1$-strong two-source extractor against classical knowledge in the Markov model.
    A quick calculation shows that the choice $k_1=k_1'+\log\frac{1}{\varepsilon'}$, $k_2=k_2'+\log\frac{1}{\varepsilon'}$, $\varepsilon=3\varepsilon'$ gives $\varepsilon = 3 \cdot 2^{-\frac{k_1+k_2+2-n-r-m}{4}}$, as desired.
\end{proof}

\section{Security of two-source extractors}
\label{chapter:results}
Up to this point, we have introduced the notion of two-source extractors and the tools that we will use for their analysis. 
Let us now come to our main results. 
Recall from the introduction that we divide our results into two approaches that we term \emph{modular} and \emph{non-modular}.
We begin with proving~\cref{prop:proof_step_1} that forms the basis of both approaches. 
Then via our modular approach, which follows in spirit the structure given in Chapter 5 of \cite{Chung_2014}, we prove~\cref{thm:general_ext_result} for generic two-source extractors. 
This is the formal version of~\cref{thm:informal:general_ext} from the introduction.
The main difference between the proof in \cite{Chung_2014} and ours lies in the use of what we call the measured XOR-Lemma, which we prove in~\cref{sec:proof_measured_XOR_lemma}.
In~\cref{sec:proof_theorem_1}, we first apply the measured XOR-Lemma to generic two-source extractors leading to~\cref{thm:general_ext_result} and then state a more refined result with improved output length that can be obtained for the extractor in \cite{Dodis_extractors}.
Restarting from~\cref{prop:proof_step_1}, we present our non-modular approach in~\cref{sec:proof_theorem_2}.
Through this approach, which mainly relies on the Fourier transform for matrix-valued functions from \cite{Fehr_Schaffner}, we prove~\cref{thm:BleA_1m_result}.
This theorem is the formal version of~\cref{thm:informal:BleA_1m} from the introduction and specific to the extractor from \cite{Dodis_extractors}.
In the final~\cref{sec:extension_to_Markov_model} we extend our results to the Markov model.

\subsection{Preparations}
In this section we present some preliminary results which will be useful for proving the statements in~\cref{sec:proof_measured_XOR_lemma,sec:proof_theorem_1,sec:proof_theorem_2}.
\subsubsection{Pretty good measurement}
\begin{definition}[Pretty good measurement (PGM)~\cite{PGM}]
\label{def:PGM}
    Let $\rho_{XB}$ be a cq-state. We define the pretty good measurement (PGM) associated to $\rho_{XB}$ as the POVM with elements
    \begin{equation*}
        \Lambda^{\rho_{XB};x}_{B} = \rho_{B}^{-1/2} \rho_{B \land x} \rho_{B}^{-1/2}.
    \end{equation*}
    We denote the channel obtained from measuring the PGM by $\Lambda^{\rho_{XB}}_{X'|B}$, i.e.
    \begin{equation*}
        \Lambda^{\rho_{XB}}_{X'|B}[\sigma_{B}] = \sum_{x} \ketbra{x}{x}_{X'} \tr{\Lambda^{\rho_{XB};x}_{B} \sigma_{B}}.
    \end{equation*}
\end{definition}

The following statement follows immediately from the definition of PGMs.
\begin{proposition}\label{prop:PGM_f_commute}
    % Let $\mathcal{H}_X$, $\mathcal{H}_Y$ and $\mathcal{H}_Z$ be Hilbert spaces associated with random variables $X,Y$ and $Z$, that have alphabets $\mathcal{X}, \mathcal{Y}$ and $\mathcal{Z}$, i.e.\ let $\{|x\rangle\}_{x\in\mathcal{X}}$, $\{|y\rangle\}_{y\in\mathcal{Y}}$ and $\{|z\rangle\}_{z\in\mathcal{Z}}$ be orthonormal bases of $\mathcal{H}_X$, $\mathcal{H}_Y$ and $\mathcal{H}_Z$, respectively.
    % Let $\rho_{XB}\in{\mathcal{S}(\mathcal{H}_X\otimes\mathcal{H}_B)}$ and $\sigma_{ZB}\in{\mathcal{S}(\mathcal{H}_Z\otimes\mathcal{H}_B)}$ be cq-states that are classical on $\mathcal{H}_X$ and $\mathcal{H}_Z$, respectively.
    Let $\rho_{XB}$ be a quantum state which is classical on $X$. Then for any function $f:\mathcal{X}\mapsto\mathcal{Y}$, we have
    \begin{equation*}
        \Lambda^{\rho_{f(X)B}}_{Y'|B} = f_{Y'|X'} \circ\Lambda^{\rho_{XB}}_{X'|B}.
    \end{equation*}
    %where $\PGM[\rho]{X}{B}{X}$ denotes the application of a PGM that has POVM-elements $\PGMelements[\rho]{X}{B}{x}=P_X(x)\rho_B^{-1/2}\rho_B^x\rho_B^{-1/2}$ and the channel $\PGMchannel{f(X)}{B}{Y'}$ corresponding to the application of a PGM that has POVM-elements $\rho_B^{-1/2}\sum_{x\;\mathrm{s.t.}\;f(x)=y} P_X(x)\rho_B^x\rho_B^{-1/2}$.
\end{proposition}

\subsubsection{Fourier transform of matrix-valued functions}\label{sec:Fourier_trafo}
Fehr and Schaffner \cite{Fehr_Schaffner} introduced the Fourier transform and L\textsubscript{2}-norm for matrix-valued functions to study an extractor family called $\delta$-biased masking.
% They showed that it can be helpful to think of some density operators, which in the case of finite-dimensional Hilbert spaces can be represented by complex-valued matrices, as the output of a matrix-valued function.
% This allows them to use an anlogue of Parseval's theorem for matrix-valued functions.
We give their definitions below and refer to \cite{Fehr_Schaffner} for more details and a proof of the following lemma.
Let $\mathcal{MF}(n)$ be the complex vector space of all matrix-valued functions $M: \{0,1\}^n \mapsto \mathbb{C}^{d\times d}$, where $n,d\in\mathbb{N}$ are fixed.

\begin{definition}\label{def:_matrix_valued_FT}
    Let $M\in\mathcal{MF}(n)$. 
    The Fourier transform of $M$ is given by the matrix-valued function
    \begin{equation*}
    \begin{aligned}
        \mathfrak{F}[M]: \{0, 1\}^{n} &\to \mathbb{C}^{d \times d} \\
        \alpha &\mapsto \frac{1}{\sqrt{2^{n}}}\sum_{x\in\{0,1\}^n}(-1)^{\alpha\cdot x}M(x).
    \end{aligned}
    \end{equation*}
\end{definition}

\begin{definition}\label{def:matrix_valued_norm}
    The L\textsubscript{2}-norm of $M\in\mathcal{MF}(n)$ is defined as 
    \begin{equation*}
        \trinorm{M}_2 \coloneq \sqrt{\tr{\sum_{x\in\{0,1\}^n} M(x)^\dagger M(x)}}.
    \end{equation*}
\end{definition}

The following lemma can be seen as an analogue of Parseval's theorem.
\begin{lemma}[{\cite[Lemma 4.1]{Fehr_Schaffner}}]\label{lem:Parseval}
    Let $M\in\mathcal{MF}(n)$. Then $\trinorm{\mathfrak{F}[M]}_2 = \trinorm{M}_2$.
\end{lemma}

\subsubsection{A useful proposition}
The following proposition will provide the starting point of both the modular and the non-modular approach.
\begin{proposition}\label{prop:proof_step_1}
    Let Z be a $m$-bit system and $E$ be an arbitrary quantum system. For any cq-state $\rho_{ZE} = \sum_{z\in\{0,1\}^m} [z]_Z \otimes \rho_{E \land z}$ and any state $\sigma_E\in\mathcal{S}(\mathcal{H}_E)$ with $\ker(\id_{Z} \otimes \sigma_{E}) \subseteq \ker(\rho_{ZE})$, we have
    \begin{equation*}
        \delta\bigl(\rho_{ZE}, \omega_{Z} \otimes \rho_E \bigr)^2 \leq \frac{1}{4} \sum_{0^m \neq s\in \{0,1\}^m}\tr{ \sum_{z,z'\in\{0,1\}^m} (-1)^{s\cdot z + s\cdot z'} M(z)M(z')},
    \end{equation*}
    where we defined the matrix-valued function $M\in\mathcal{MF}(m)$ as $M(z) \coloneqq \sigma_E^{-1/4} \rho_{E \land z} \sigma_E^{-1/4}$.
\end{proposition}
\begin{proof}
    By definition of the Fourier transform, we have
    \begin{equation}
        \label{eq:FT_M_0}
        \mathfrak{F}[M](0^m) = \frac{1}{\sqrt{2^m}}\sum_{z\in\{0,1\}^m} M(z) = \frac{1}{\sqrt{2^m}}\sigma_E^{-1/4}\rho_E\sigma_E^{-1/4}.
    \end{equation}
    In the following, $z$, $z'$, and $s$ will denote $m$-bit strings. Using~\cref{cor:bound_one_norm_by_two_norm} and~\cref{lem:d_2_cq_states}, we find
    \begin{equation}
        \begin{split}
            \delta\big(\rho_{ZE}, \omega_{Z} \otimes \rho_E \big)^2 &\leq \frac{1}{4} \mathrm{dim}(\mathcal{H}_Z)\tr{\sigma_E}d_2(\rho_{ZE}|\sigma_E) \\
                & = \frac{2^m}{4} d_2(\rho_{ZE}|\sigma_E)\\
                & = \frac{2^m}{4} \mleft( \sum_{z} \tr{\mleft( \sigma_E^{-1/4} \rho_{E \land z} \sigma_E^{-1/4} \mright)^2 } - \frac{1}{2^m} \tr{\mleft( \sigma_E^{-1/4} \rho_E \sigma_E^{-1/4} \mright)^2 } \mright)\\
                & = \frac{2^m}{4} \mleft( \sum_{z} \tr{ M(z)^\dagger M(z) } - \tr{ \mathfrak{F}[M](0^m)^2 } \mright),
        \end{split}
    \end{equation}
    where in the last line we used the definition of $M(z)$ and~\cref{eq:FT_M_0}. 
    Continuing from here, we obtain
    \begin{equation}
        \begin{split}
            \delta\big(\rho_{ZE},\omega_{Z}\otimes \rho_E \big)^2 
                & \leq \frac{2^m}{4} \mleft(\trinorm{M}_2^2 - \tr{ \mathfrak{F}[M](0^m)^2 } \mright)\\
                & = \frac{2^m}{4} \mleft(\trinorm{\mathfrak{F}[M]}_2^2 - \tr{ \mathfrak{F}[M](0^m)^2 } \mright)\\
                & = \frac{2^m}{4} \mleft(\sum_{s} \tr{\mathfrak{F}[M](s)^\dagger \mathfrak{F}[M](s) }  - \tr{\mathfrak{F}[M](0^m)^2 } \mright)\\
                & = \frac{2^m}{4} \sum_{s\neq 0^m} \tr{\mathfrak{F}[M](s)^2 } \\
                &  = \frac{1}{4} \sum_{s\neq 0^m} \tr{ \sum_{z,z'} (-1)^{s\cdot z}(-1)^{s\cdot z'} M(z)M(z') },        
        \end{split}
    \end{equation}
    where the first equality follows by~\cref{lem:Parseval}. All other equalities follow by~\cref{def:_matrix_valued_FT} of the Fourier transform and~\cref{def:matrix_valued_norm} of the L\textsubscript{2}-norm for matrix-valued functions.
\end{proof}

\subsection{The Measured XOR-Lemma}\label{sec:proof_measured_XOR_lemma}
Building on~\cref{prop:proof_step_1} we find what we call the \textit{measured XOR-Lemma}, inspired by \cite{Kasher_Kempe_main} and \cite{Vazirani}. 
%\todo{should we add Goldreich ?}
It constitutes the main tool used in the modular approach presented in~\cref{sec:proof_theorem_1}.

\begin{lemma}[Measured XOR-Lemma]\label{lem:measured_XOR_lemma}
    Let $Z$ be a $m$-bit system and let $E$ be an arbitrary quantum system.
    For any cq-state $\rho_{ZE} = \sum_{z\in\{0,1\}^m} [z]_Z\otimes \rho_{E \land z}$, we have
    % \begin{equation*}
    %     \delta\bigl(\rho_{ZE}, \omega_{Z} \otimes \rho_E \bigr) \leq \sqrt{\frac{1}{2} \sum_{0\neq s \in \{0,1\}^m} \delta\bigl( \rho_{(s\cdot Z) \PGM[\rho]{(s\cdot Z)}{E}{\TLS'}} , \omega_{1}\otimes \rho_{\PGM[\rho]{(s\cdot Z)}{E}{\TLS'}} \bigr)},
    % \end{equation*}
    \begin{equation*}
        \delta\bigl(\rho_{ZE}, \omega_{Z}\otimes\rho_E \bigr) \leq \sqrt{\frac{1}{2} \sum_{0\neq s \in \{0,1\}^m} \delta\left( \PGMchannelMartin{\rho_{(s \cdot Z)E}}{\TLS|E}[\rho_{(s\cdot Z) E}] , \omega_{1} \otimes \PGMchannelMartin{\rho_{(s \cdot Z) E}}{\TLS|E}[\rho_{E}] \right)},
    \end{equation*}
    where $\TLS$ is a 1-bit system.
\end{lemma}
% Note that~\cref{lem:measured_XOR_lemma} implies the classical-quantum XOR-Lemma in \cite[Lemma 10]{Kasher_Kempe_main} due to the data-processing inequality.
\begin{proof}
%    Let $\mathrm{dim}(\mathcal{H}_E) = D =2^d$.
    In the rest of this proof $z$, $z'$, and $s$ will denote $m$-bit strings and $i,i'$ will denote single bits.
    Let us define 
    \begin{equation}\label{eq:qubit_knowledge}
        \rho_{E \land (s\cdot Z = i)} \coloneq \sum_{z\;\mathrm{s.t.}\; s\cdot z = i} \rho_{E \land z},
    \end{equation}
    which can be seen as the (unnormalized) conditional operators of the state $\rho_{(s\cdot Z)E}$.
    % In contrast to~\cref{def:cq_state}, however, $\rho_E^{(s\cdot Z = i)}$ are non-normalized conditional operators, since they include the probability $P_{s\cdot Z}(i)$.
    From~\cref{prop:proof_step_1} with $\sigma_{E} = \rho_{E}$, recall $M(z) = \rho_E^{-1/4} \rho_{E \land z} \rho_E^{-1/4}$ and define
    \begin{equation}
        M_s(i) \coloneq \sum_{z\;\mathrm{s.t.}\; s \cdot z = i} M(z) = \rho_E^{-1/4} \rho_{E \land (s\cdot Z=i)} \rho_E^{-1/4}.
    \end{equation}
    %alternative sum: \substack{z\;\mathrm{s.t.}\\s\cdot z = i}
    
    With these definitions at hand, let us continue with the inequality that we found in~\cref{prop:proof_step_1}
    \begin{equation}\label{eq:proof:XOR_lemma:applying_prop_proof_step_1}
        \begin{split}
            \delta\bigl(\rho_{ZE}, \omega_{Z}\otimes\rho_E \bigr)^2 
            &\leq \frac{1}{4} \sum_{s\neq 0^m}\tr{ \sum_{z,z'} (-1)^{s\cdot z + s\cdot z'} M(z)M(z')}\\
            &= \frac{1}{4} \sum_{s\neq 0^m}\tr{ \sum_{i,i'} (-1)^{i + i'} M_s(i)M_s(i')}\\
            &= \frac{1}{4} \sum_{s\neq 0^m} \sum_{i,i'} (-1)^{i + i'} \tr{\rho_{E \land (s\cdot Z=i)} \rho_E^{-1/2} \rho_{E \land (s\cdot Z = i')} \rho_E^{-1/2} }\\
            &= \frac{1}{4} \sum_{s\neq 0^m} \sum_{i,i'} (-1)^{i + i'} \tr{\rho_{E \land (s\cdot Z=i')} \PGMelements[]{(s\cdot Z)}{E}{i} },
        \end{split}
    \end{equation}
    where we identified $M_s(i)$ and $M_s(i')$ in the first equality.
    For the second equality we used linearity and cyclicity of the trace.
    % Since $\sigma_E$ is arbitrary according to~\cref{prop:proof_step_1}, we made the choice $\sigma_E = \rho_E$ in the last equality.
    We then identified the POVM-elements $\PGMelements[]{(s\cdot Z)}{E}{i} = \rho_E^{-1/2}\rho_{E \land (s\cdot Z = i)}\rho_E^{-1/2}$ of a PGM.
    Note, that we omitted the superscript $\rho$ for the POVM-elements as this will not change throughout the proof.
%\todo{Note: Does it make sense to introduce $M_s(i)$ if we do not use it anywhere else in the proof?}
    
    Focusing on the sum over $i,i'$ in the last step of~\cref{eq:proof:XOR_lemma:applying_prop_proof_step_1}, we obtain 
    \begin{equation}\label{eq:proof:XOR_lemma:back_to_trace_norm}
        \begin{split}
        \sum_{i,i'} (-1)^{i+i'} \tr{ \rho_{E \land (s \cdot Z=i')} \PGMelements[]{(s\cdot Z)}{E}{i} }
            &= \sum_{i,i'} (-1)^{i+i'} \mleft( \tr{ \rho_{E \land (s\cdot Z=i')} \PGMelements[]{(s\cdot Z)}{E}{i} } - \frac{1}{2} \tr{\rho_E \PGMelements[]{(s\cdot Z)}{E}{i}} \mright) \\
            &\leq \sum_{i,i'} \abs{ \tr{ \left( \rho_{E \land (s\cdot Z=i')} - \frac{1}{2}\rho_E \right) \PGMelements[]{(s\cdot Z)}{E}{i} } }\\
            &= \norm{ \sum_{i,i'} [i]_\TLS\otimes[i']_{\TLS'} \tr{ \mleft(\rho_{E \land (s\cdot Z = i')} - \frac{1}{2}\rho_E \mright) \PGMelements[]{(s\cdot Z)}{E}{i} }  }_1.
        \end{split}
    \end{equation}
    For the first equality note that $\sum_{i'}(-1)^{i+i'}\tr{\rho_E \PGMelements[]{(s\cdot Z)}{E}{i}} = 0$. 
    The inequality follows because any real number is upper-bounded by its absolute value.
    Note, that the resulting expression can be interpreted as the trace norm of a diagonal operator that has $\tr{ \left( \rho_{E \land (s\cdot Z=i')} - \frac{1}{2}\rho_E \right) \PGMelements[]{(s\cdot Z)}{E}{i} }$ as its diagonal entries.
    This gives the last equality.
    We immediately see that 
    % \begin{equation}
    %     \rho_{(s\cdot Z) \PGM[]{(s\cdot Z)}{E}{\TLS'}} = \PGMchannel[]{(s\cdot Z)}{E}{\TLS'}(\rho_{(s\cdot Z)E}) = \sum_{i,i'} [i]_\TLS\otimes[i']_{\TLS'} \tr{\rho_E^{(s\cdot Z) = i} \PGMelements[]{(s\cdot Z)}{E}{i'} }
    % \end{equation}
    \begin{equation}
        \PGMchannelMartin{\rho_{(s \cdot Z) E}}{\TLS|E}[\rho_{(s \cdot Z) E}] = \sum_{i,i'} [i]_\TLS\otimes[i']_{\TLS'} \tr{\rho_{E \land (s\cdot Z = i')} \PGMelements[]{(s\cdot Z)}{E}{i} }
    \end{equation}
    and
    % \begin{equation}
    %     \omega_{1} \otimes \rho_{\PGM[]{(s\cdot Z)}{E}{\TLS'}} = \omega_{1} \otimes \PGMchannel[]{(s\cdot Z)}{E}{\TLS'}(\rho_E) = \sum_{i,i'} \frac{1}{2} [i]_\TLS\otimes[i']_{\TLS'} \tr{ \rho_E \PGMelements[]{(s\cdot Z)}{E}{i'} },
    % \end{equation}
    \begin{equation}
        \omega_{1} \otimes \PGMchannelMartin{\rho_{(s \cdot Z) E}}{\TLS|E}[\rho_{E}] = \sum_{i,i'} \frac{1}{2} [i]_\TLS\otimes[i']_{\TLS'} \tr{ \rho_E \PGMelements[]{(s\cdot Z)}{E}{i} }.
    \end{equation}
    Plugging this into~\cref{eq:proof:XOR_lemma:back_to_trace_norm} gives 
    % \begin{equation}
    %     \sum_{i,i'} (-1)^{i+i'} \tr{ \rho_E^{(s\cdot Z)=i} \PGMelements[]{(s\cdot Z)}{E}{i'} } 
    %     \leq 2\delta\left( \rho_{(s\cdot Z) \PGM[]{(s\cdot Z)}{E}{\TLS'}}, \omega_{1}\otimes \rho_{\PGM[]{(s\cdot Z)}{E}{\TLS'}} \right).
    % \end{equation}
    \begin{equation}
        \sum_{i,i'} (-1)^{i+i'} \tr{ \rho_{E \land (s\cdot Z=i')} \PGMelements[]{(s\cdot Z)}{E}{i} } 
        \leq 2\delta\left( \PGMchannelMartin{\rho_{(s \cdot Z)E}}{\TLS|E}[\rho_{(s\cdot Z) E}], \omega_{1} \otimes \PGMchannelMartin{\rho_{(s \cdot Z)E}}{\TLS|E}[\rho_{E}] \right).
    \end{equation}
    Together with~\cref{eq:proof:XOR_lemma:applying_prop_proof_step_1}, this gives the desired bound.
    %\todo{Note: is there anything useful we could find from directly using data-processing in order to lower bound delta(rho ZE,...) and show tightness or something similar}
\end{proof}

Applying the measured XOR-Lemma to the definition of a strong two-source extractor gives the following result.
\begin{corollary}\label{cor:measured_XOR_strong_ext}
    Let $\rho_{X_1 X_2 C_1 C_2} = \rho_{X_1 C_1}\otimes\rho_{X_2 C_2}$, where $\rho_{X_1 C_1}$ and $\rho_{X_2 C_2}$ are cq-states. For any function $\Ext: \{0,1\}^{n_1}\times\{0,1\}^{n_2}\mapsto\{0,1\}^m$, we have
    % \begin{equation}
    %     \begin{split}
    %         \delta\big(&\rho_{\Ext(X_1,X_2)X_1 C}, \rho_{U_m}\otimes\rho_{X_1 C_1}\otimes\rho_{C_2}\big)\\
    %         &\leq \sqrt{ \frac{1}{2} \sum_{0\neq s \in \{0,1\}^m} \delta\left( \rho_{(s\cdot \Ext(X_1,X_2)) X_1 \PGM[]{X_2}{C_2}{X_2'}} ,  \rho_{U_1}\otimes \rho_{X_1} \otimes \rho_{\PGM[]{X_2}{C_2}{X_2'}} \right) }. \\
    %     \end{split}
    % \end{equation}
    \begin{equation}
        \begin{split}
            \delta\big(&\rho_{\Ext(X_1,X_2) X_1 C_1 C_2}, \omega_{m} \otimes \rho_{X_1 C_1} \otimes \rho_{C_2}\big)\\
            &\leq \sqrt{ \frac{1}{2} \sum_{0\neq s \in \{0,1\}^m} \delta\left( \PGMchannelMartin{\rho_{X_2 C_2}}{X'_{2}|C_2}[\rho_{(s \cdot \Ext(X_1,X_2)) X_1 C_2}], \omega_{1} \otimes \rho_{X_1} \otimes \PGMchannelMartin{\rho_{X_2 C_2}}{X'_2|C_2}[\rho_{C_2}] \right) }. \\
        \end{split}
    \end{equation}
\end{corollary}
A similar result can be obtained for the security condition of an $X_2$-strong two-source extractor with reversed roles $X_1 C_1 \leftrightarrow X_2 C_2$.

\begin{proof}
    %\begin{equation}\label{eq:Ext_state}
    %    \rho_{\Ext(x_1, X_2)C_2} 
    %    = \sum_{z\in\{0,1\}^m} [z] \otimes \rho_{C_2}^{(\Ext(x_1,X_2)=z)}
    %\end{equation}
    %where we defined the non-normalized conditional operator
    %\begin{equation}
    %    \rho_{C_2}^{(\Ext(x_1,X_2)=z)} \coloneq \sum_{\substack{x_2\in\{0,1\}^n \;\mathrm{s.t.} \\ \Ext(x_1,x_2)=z }} P_{X_2}(x_2) \rho_{C_2}^{x_2}.
    %\end{equation}
    %Then we have $\PGMelements[]{\Ext(x_1,X_2)}{C_2}{z} = \rho_{C_2}^{-1/2} \rho_{C_2}^{(\Ext(x_1,X_2)=z)} \rho_{C_2}^{-1/2}$
    
    In the rest of this proof, $s$ will denote $m$-bit strings and $x_1$ will denote $n_1$-bit strings.
    % In the interest of clarity, we will omit the sets $\{0,1\}^m$ and $\{0,1\}^{n_1}$ when writing down sums.
    % We have
    % \begin{equation}\label{eq:proof:cor:measured_XOR:applying_measured_XOR}
    %     \begin{split}
    %         \delta\big(&\rho_{\Ext(X_1,X_2)X_1 C}, \omega_{m}\otimes\rho_{X_1 C_1}\otimes\rho_{C_2}\big) \\
    %         &= \mathop{\mathbb{E}}_{x_1\leftarrow X_1} \mleft[\delta\big(\rho_{\Ext(x_1,X_2)C_2}, \omega_{m}\otimes\rho_{C_2}\big)\mright] \\
    %         &\leq \mathop{\mathbb{E}}_{x_1\leftarrow X_1} \mleft[ \sqrt{\frac{1}{2} \sum_{s\neq0^m} \delta\bigl( \rho_{(s\cdot \Ext(x_1,X_2)) \PGM[]{(s\cdot\mathrm{Ext}(x_1,X_2))}{C_2}{\TLS'}}, \omega_{1}\otimes \rho_{\PGM[]{(s\cdot\mathrm{Ext}(x_1,X_2))}{C_2}{\TLS'}} \bigr)} \mright]  \\
    %         &= \mathop{\mathbb{E}}_{x_1\leftarrow X_1} \mleft[ \sqrt{\frac{1}{2} \sum_{s\neq0^m} \delta\bigl( \rho_{(s\cdot \Ext(x_1,X_2)) (s\cdot\mathrm{Ext}(x_1,\PGM[]{X_2}{C_2}{X_2'}))}, \omega_{1}\otimes \rho_{(s\cdot\mathrm{Ext}(x_1,\PGM[]{X_2}{C_2}{X_2'}))} \bigr)} \mright]  \\
    %         &\leq \mathop{\mathbb{E}}_{x_1\leftarrow X_1} \mleft[ \sqrt{\frac{1}{2} \sum_{s\neq0^m} \delta\bigl( \rho_{(s\cdot \Ext(x_1,X_2)) \PGM[]{X_2}{C_2}{X_2'}}, \omega_{1}\otimes \rho_{\PGM[]{X_2}{C_2}{X_2'}} \bigr)} \mright] \\
    %     \end{split}  
    % \end{equation}
    For readability, let us define $f^{s,x_1}_{}(x_2) \coloneqq s \cdot \mathrm{Ext}(x_1, x_2)$. We have
    \begin{equation}\label{eq:proof:cor:measured_XOR:applying_measured_XOR}
        \begin{split}
            \delta\big(&\rho_{\Ext(X_1,X_2)X_1 C}, \omega_{m}\otimes\rho_{X_1 C_1}\otimes\rho_{C_2}\big) \\
            &= \mathop{\mathbb{E}}_{x_1\leftarrow X_1} \mleft[\delta\big(\rho_{\Ext(x_1,X_2)C_2}, \omega_{m}\otimes\rho_{C_2}\big)\mright] \\
            &\leq \mathop{\mathbb{E}}_{x_1\leftarrow X_1} \mleft[ \sqrt{\frac{1}{2} \sum_{s\neq0^m} \delta\left( \PGMchannelMartin{\rho_{f^{s,x_1}(X_2) C_2}}{\TLS'|C_{2}}[\rho_{(s \cdot \mathrm{Ext}(x_1, X_2)) C_2}], \omega_{1} \otimes \PGMchannelMartin{\rho_{f^{s,x_1}(X_2) C_2}}{\TLS'|C_{2}}[\rho_{C_2}] \right)} \mright]  \\
            &= \mathop{\mathbb{E}}_{x_1\leftarrow X_1} \mleft[ \sqrt{\frac{1}{2} \sum_{s\neq0^m} \delta\left( f^{s,x_1}_{\TLS'|X'_2} \circ \PGMchannelMartin{\rho_{X_2 C_2}}{X'_2|C_2}[\rho_{(s \cdot \mathrm{Ext}(x_1, X_2)) C_2}], \omega_{1} \otimes f^{s,x_1}_{\TLS'|X'_2} \circ \PGMchannelMartin{\rho_{X_2 C_2}}{X'_2|C_2}[\rho_{C_2}] \right)} \mright]  \\
            &\leq \mathop{\mathbb{E}}_{x_1\leftarrow X_1} \mleft[ \sqrt{\frac{1}{2} \sum_{s\neq0^m} \delta\left( \PGMchannelMartin{\rho_{X_2 C_2}}{X'_2|C_2}[\rho_{(s \cdot \mathrm{Ext}(x_1, X_2)) C_2}], \omega_{1} \otimes \PGMchannelMartin{\rho_{X_2 C_2}}{X'_2|C_2}[\rho_{C_2}] \right)} \mright]  \\
        \end{split}  
    \end{equation}
    where we used~\cref{eq:strong_ext_condition} in the first equality.
    In line two, we applied the measured XOR-Lemma with $Z=\Ext(x_1,X_2)$ and $E=C_2$.
    % Recall, that in the Measured XOR-Lemma $\TLS'$ denotes a two-level system.
    In line three, we used~\cref{prop:PGM_f_commute} and in line four data-processing of the trace norm. % in order to get from an expression with $\PGM[]{(s\cdot\mathrm{Ext}(x_1,X_2))}{C_2}{\TLS'}$ to an expression with $\PGM[]{X_2}{C_2}{X_2'}$. 
    % Thereby, we remove the implicit dependence of $\PGM[]{(s\cdot\mathrm{Ext}(x_1,X_2))}{C_2}{\TLS'}$ on $x_1$. 
    
    % We recall Jensen's inequality that states $\mathop{\mathbb{E}}_{i}[f(g_i)]\leq f(\mathop{\mathbb{E}}_{i}[g_i])$ for any concave function $f$ and $g_i\in\mathbb{R}$.
    % Applying it to the equation above, we find
    Applying Jensen's inequality gives
    % \begin{equation}
    %     \begin{split}
    %         \delta\big(&\rho_{\Ext(X_1,X_2)X_1 C}, \omega_{m}\otimes\rho_{X_1 C_1}\otimes\rho_{C_2}\big) \\
    %          &\leq \sqrt{ \mathop{\mathbb{E}}_{x_1\leftarrow X_1} \mleft[ \frac{1}{2} \sum_{s\neq0^m} \delta\bigl( \rho_{(s\cdot \Ext(x_1,X_2)) \PGM[]{X_2}{C_2}{X_2'}} , \omega_{1}\otimes \rho_{\PGM[]{X_2}{C_2}{X_2'}} \bigr) \mright] } \\
    %          &= \sqrt{ \frac{1}{2} \sum_{s\neq0^m} \delta\bigl( \rho_{(s\cdot \Ext(X_1,X_2)) X_1 \PGM[]{X_2}{C_2}{X_2'}} , \omega_{1}\otimes \rho_{X_1} \otimes \rho_{\PGM[]{X_2}{C_2}{X_2'}} \bigr) },
    %     \end{split}
    % \end{equation}
    \begin{equation}
        \begin{split}
            \delta\big(&\rho_{\Ext(X_1,X_2)X_1 C_1 C_2}, \omega_{m}\otimes\rho_{X_1 C_1}\otimes\rho_{C_2}\big) \\
            &\leq \sqrt{ \mathop{\mathbb{E}}_{x_1\leftarrow X_1} \mleft[ \frac{1}{2} \sum_{s\neq0^m} \delta\left( \PGMchannelMartin{\rho_{X_2 C_2}}{X'_2|C_2}[\rho_{(s \cdot \mathrm{Ext}(x_1, X_2) C_2}], \omega_{1} \otimes \PGMchannelMartin{\rho_{X_2 C_2}}{X'_2|C_2}[\rho_{C_2}] \right) \mright]} \\
            &= \sqrt{ \frac{1}{2} \sum_{s\neq0^m} \delta\left( \PGMchannelMartin{\rho_{X_2 C_2}}{X'_2|C_2}[\rho_{(s \cdot \mathrm{Ext}(X_1, X_2) X_1 C_2}], \omega_{1} \otimes \rho_{X_1} \otimes \PGMchannelMartin{\rho_{X_2 C_2}}{X'_2|C_2}[\rho_{C_2}] \right)},
        \end{split}
    \end{equation}
    where in the last line we used~\cref{eq:strong_ext_condition} again.
\end{proof}

\subsection{Security of generic extractors}\label{sec:proof_theorem_1}
With the previous results, we can state our first main result.
\begin{theorem}\label{thm:general_ext_result}
    Any $(k_1, k_2, \varepsilon)$ $X_1$-strong two-source extractor against classical product-type knowledge is a $(k_1, k_2, \sqrt{2^{m-1} \varepsilon})$ $X_1$-strong two-source extractor against quantum product-type knowledge.
\end{theorem}
\begin{proof}
    As before, $s$ will denote $m$-bit strings.
    Let $\mathrm{Ext}: \{0,1\}^{n_1} \times \{0,1\}^{n_2} \rightarrow \{0,1\}^m$ be a $(k_1, k_2, \varepsilon)$ $X_1$-strong two-source extractor. 
    Let $X_1$, $X_2$ be randomness sources with quantum side information $C = C_1 C_2$ of the form $\rho_{X_1 X_2 C}=\rho_{X_1 C_1}\otimes\rho_{X_2 C_2}$, satisfying $H_{\min}(X_1|C_1)\geq k_1$ and $H_{\min}(X_2|C_2)\geq k_2$.
    From~\cref{cor:measured_XOR_strong_ext} we know that 
    % $\mathcal{D}_{X_2'|C_2}$ denotes the channel that applies the PGM $\{\mathcal{D}^{x_2}\}_{x_2\in\{0,1\}^n}$ with POVM-elements $\mathcal{D}^{x_2} \coloneq P_{X_2}(x_2) \rho_{C_2}^{-1/2} \rho_{C_2}^{x_2} \rho_{C_2}^{-1/2}$ and stores the results in $X_2'$.
    % \begin{equation}\label{eq:proof:general_ext:data_processing}
    %     \begin{split}
    %         \delta\big(&\rho_{\Ext(X_1,X_2)X_1 C}, \omega_{m}\otimes\rho_{X_1 C_1}\otimes\rho_{C_2}\big) \\
    %         &\leq \sqrt{ \frac{1}{2} \sum_{s\neq0^m} \delta\bigl( \rho_{(s\cdot \Ext(X_1,X_2)) X_1 \PGM[]{X_2}{C_2}{X_2'}} ,  \omega_{1}\otimes \rho_{X_1} \otimes \rho_{\PGM[]{X_2}{C_2}{X_2'}} \bigr) } \\
    %         &\leq \sqrt{ \frac{1}{2} \sum_{s\neq0^m} \delta\bigl( \rho_{\Ext(X_1,X_2) X_1 \PGM[]{X_2}{C_2}{X_2'}} , \omega_{m}\otimes \rho_{X_1} \otimes \rho_{\PGM[]{X_2}{C_2}{X_2'}} \bigr) } \\
    %         &= \sqrt{ 2^{m-1} \delta\bigl( \rho_{\Ext(X_1,X_2) X_1 \PGM[]{X_2}{C_2}{X_2'}}, \omega_{m}\otimes \rho_{X_1} \otimes \rho_{\PGM[]{X_2}{C_2}{X_2'}} \bigr) } \\
    %     \end{split}
    % \end{equation}
    \begin{equation}\label{eq:proof:general_ext:data_processing}
        \begin{split}
            \delta\big(&\rho_{\Ext(X_1,X_2)X_1 C_1 C_2}, \omega_{m}\otimes\rho_{X_1 C_1}\otimes\rho_{C_2}\big) \\
            &\leq \sqrt{ \frac{1}{2} \sum_{s\neq0^m} \delta \left( \PGMchannelMartin{\rho_{X_2 C_2}}{X'_2|C_2}[\rho_{(s \cdot \mathrm{Ext}(X_1, X_2)) X_1 C_2}],  \omega_{1} \otimes \rho_{X_1} \otimes \PGMchannelMartin{\rho_{X_2 C_2}}{X'_2|C_2}[\rho_{C_2}] \right) } \\
            &\leq \sqrt{ \frac{1}{2} \sum_{s\neq0^m} \delta \left( \PGMchannelMartin{\rho_{X_2 C_2}}{X'_2|C_2}[\rho_{\mathrm{Ext}(X_1, X_2) X_1 C_2}], \omega_{m} \otimes \rho_{X_1} \otimes \PGMchannelMartin{\rho_{X_2 C_2}}{X'_2|C_2}[\rho_{C_2}] \right) } \\
            &\leq \sqrt{ 2^{m-1} \delta \left( \PGMchannelMartin{\rho_{X_2 C_2}}{X'_2|C_2}[\rho_{\mathrm{Ext}(X_1, X_2) X_1 C_2}], \omega_{m} \otimes \rho_{X_1} \otimes \PGMchannelMartin{\rho_{X_2 C_2}}{X'_2|C_2}[\rho_{C_2}] \right) } \\
        \end{split}
    \end{equation}
    where we used data-processing of the trace norm to obtain the second inequality. 
    Since the summands are independent of the bit string $s$, we obtain the third inequality.
    
    Due to the data-processing inequality of the min-entropy we have, that $H_{\mathrm{min}}(X_1) \geq H_{\mathrm{min}}(X_1|C_1) \geq k_1$ and $H_{\mathrm{min}}(X_2|X_2')_{\Lambda[\rho]} \geq H_{\mathrm{min}}(X_2|C_2)_{\rho} \geq k_2$. Since $\mathrm{Ext}$ is a $(k_1, k_2, \varepsilon)$-extractor against classical product-type knowledge we have that
    % From~\cref{prop:ext_to_ext_against_classical_product_type} we know that $\mathrm{Ext}$ is a $(k_1, k_2 + \log\frac{1}{\varepsilon}, 2\varepsilon)$ $X_1$-strong two-source extractor against classical product-type knowledge, i.e.
    % \begin{equation}
    %     \delta \bigl( \rho_{\Ext(X_1,X_2) X_1 \PGM[]{X_2}{C_2}{X_2'}} ,  \rho_{U_m}\otimes \rho_{X_1} \otimes \rho_{\PGM[]{X_2}{C_2}{X_2'}} \bigr) \leq 2\varepsilon,
    % \end{equation}
    \begin{equation*}
            \delta \left( \PGMchannelMartin{\rho_{X_2 C_2}}{X'_2|C_2}[\rho_{\mathrm{Ext}(X_1, X_2) X_1 C_2}], \omega_{m} \otimes \rho_{X_1} \otimes \PGMchannelMartin{\rho_{X_2 C_2}}{X'_2|C_2}[\rho_{C_2}] \right) \leq \varepsilon
    \end{equation*}
    which, together with the previous equation, gives the desired inequality.
\end{proof}

We remark that a similar statement can be achieved for weak extractors. For a proof, see~\cref{sec:weak_ext_proof}.
Combining the previous result with~\cref{prop:ext_to_ext_against_classical_product_type}, we get the following statement.

\begin{corollary} \label{cor:no_side_info_to_quantum_product}
    Any $(k_1, k_2, \varepsilon)$ $X_1$-strong two-source extractor is a $(k_1, k_2 + \log \frac{1}{\varepsilon}, \sqrt{2^{m} \varepsilon})$ $X_1$-strong two-source extractor against quantum product-type knowledge.
\end{corollary}

Note, that this result can be seen as a slightly improved version of \cite[corollary 16]{arnonfriedman_extractors} in the case of strong extractors (although one could obtain an equally strong version using the techniques of \cite{arnonfriedman_extractors} by using~\cref{prop:ext_to_ext_against_classical_product_type} instead of their lemma 6).

Applying~\cref{cor:no_side_info_to_quantum_product} to the \textDEOR-extractor, we find that \textDEOR{} remains secure with security parameter $\varepsilon = 2^{-(k_1 + k_2 + 2 - n - r - 4m)/6}$. This is significantly worse than the error against no side information in~\cref{lem:BleA_ext_no_side_info}.

Due to the modularity of our proof, which is based on the measured XOR-Lemma, we can further improve~\cref{thm:general_ext_result} in the case of specific extractors.
The reason for this is that the data-processing inequality in~\cref{eq:proof:general_ext:data_processing} may be very loose. If we know that the error for the extractor $s \cdot \mathrm{Ext}$ is significantly smaller than the error for $\mathrm{Ext}$, we can get an improved bound.
We demonstrate this for the case of the \textDEOR{}-extractor, where we can use the fact that $s\cdot\DEOR(X_1,X_2) = \mathrm{IP}(X_1, A_s^T X_2)$ which, together with two-universality of the inner product $\mathrm{IP}$, gives the following result:
\begin{restatable}{lemma}{bleAres}
    \label{lem:BleA_2m_result}
    The function \textDEOR{} is a $(k_1,k_2,2^{-\frac{k_1+k_2+3-n-r-2m}{4}})$ $X_1$-strong two-source extractor against quantum product-type knowledge.
\end{restatable}

The detailed proof of this result can be found in~\cref{app:proof:lem:BleA_2m_result}.
Note, that the error $\varepsilon$ in~\cref{lem:BleA_2m_result} is still larger than the error in~\cref{lem:BleA_ext_no_side_info}, where no side information is considered. 
Let us focus on closing this gap next.

\subsection{Security of the Dodis extractor}\label{sec:proof_theorem_2}
In this section we present our second, non-modular approach that gives us our second main result.
\begin{theorem}\label{thm:BleA_1m_result}
    The \textDEOR{}-extractor from~\cref{def:BleA_ext} is a $(k_1,k_2,\varepsilon)$ $X_1$-strong extractor against quantum product-type knowledge, where
    \begin{equation*}
        \varepsilon = 2^{-\frac{k_1+k_2+2-n-r-m}{2}}.
    \end{equation*}
\end{theorem}

\begin{proof}
    In the following, $x$ will represent $n$-bit strings whereas $s, z$ will represent $m$-bit strings.
    % When writing sums we will omit the sets $\{0,1\}^n$ and $\{0,1\}^m$.
    We consider the \textDEOR-extractor in the setting of strong extractors against product-type quantum side information.
    Let $X_1$, $X_2$ be randomness sources with quantum side information $C=C_{1}C_{2}$ of the form $\rho_{X_1 X_2 C_1 C_2}=\rho_{X_1 C_1}\otimes\rho_{X_2 C_2}$, satisfying $H_{\min}(X_1|C_1)\geq k_1$ and  $H_{\min}(X_2|C_2)\geq k_2$.
    To show that the \textDEOR{}-extractor is $X_1$-strong we consider the state 
    \begin{equation}
        \begin{split}
        \rho_{\DEOR(X_1,X_2) X_1 C} & = ((\DEOR)_{Z|X_1 X_2}\otimes\mathcal{I}_{X_1 C})[\rho_{X_1 X_2 X_1 C}]\\
            & = \sum_{x_1,x_2} [\DEOR(x_1,x_2)]_{Z}\otimes[x_1]\otimes\rho_{C_1 \land x_1}\otimes\rho_{C_2 \land x_2}\\
            & = \sum_{x_1,x_2} [\DEOR(x_1,x_2)]_Z \otimes \rho_{X_1 C \land x_1, x_2},
        \end{split}
    \end{equation}    
    where we introduced the non-normalized state $\rho_{X_1 C \land x_1, x_2} = [x_1]_{X_1} \otimes \rho_{C_1 \land x_1}\otimes\rho_{C_2 \land x_2}$.
%    with the second $X_1$ being an additional register containing a copy of the first $X_1$ register.
    From~\cref{prop:proof_step_1} with $Z=\DEOR(X_1,X_2)$ and $E=X_1 C_1 C_2$, we obtain
    \begin{equation}\label{eq:M_BleA}
        M(z) = \sum_{\substack{x_1,x_2\;\mathrm{s.t.} \\\DEOR(x_1,x_2) = z} } \sigma_{X_1 C}^{-1/4}\rho_{X_1 C \land x_1,x_2}\sigma_{X_1 C}^{-1/4}
    \end{equation}
    and
    \begin{equation}\label{eq:applying_prop_proof_step_1_in_reaching_classical}
        \begin{split}
            \delta&\bigl(\rho_{\DEOR(X_1,X_2) X_1 C}, \omega_{m}\otimes\rho_{X_1 C} \bigr)^2 \\
                & \leq \frac{1}{4} \sum_{s\neq 0^m}\tr{ \sum_{z,z'} (-1)^{s\cdot z + s\cdot z'} M(z)M(z') }\\    
                & = \frac{1}{4} \sum_{s\neq 0^m} \sum_{x_1,x_2,x_1',x_2'} (-1)^{s\cdot \DEOR(x_1,x_2) + s\cdot \DEOR(x_1',x_2')} \\
                &\hspace{3cm}\tr{\mleft(\sigma_{X_1 C}^{-1/4}\rho_{X_1 C \land x_1,x_2}\sigma_{X_1 C}^{-1/4}\mright) \mleft(\sigma_{X_1 C}^{-1/4}\rho_{X_1 C \land x_1',x_2'}\sigma_{X_1 C}^{-1/4}\mright)},
        \end{split}
    \end{equation}
    where in the third line we plugged in~\cref{eq:M_BleA} and used linearity of the trace. 
    Note, that the inequality holds for arbitrary $\sigma_{X_1 C}$ satisfying the assumption of~\cref{prop:proof_step_1}.
    % \jakob{"almost arbitrary"?} \carla{This refers to the ker condition, right?}
    For simplicity, we choose $\sigma_{X_1 C} = \rho_{X_1 C_1} \otimes \sigma_{C_2}$ for some $\sigma_{C_2}$ to be specified later. 
    Since $\rho_{X_1 C_1}$ is a cq-state, we have $\rho_{X_1 C_1}^{-1/4} = \sum_{x_1}[x_1]_{X_1} \otimes(\rho_{C_1 \land x_1})^{-1/4}$. 
    Focusing on the first term inside the trace, we find
    \begin{equation}
        \begin{split}
            \sigma_{X_1 C}^{-1/4}\rho_{X_1 C \land x_1,x_2}\sigma_{X_1 C}^{-1/4} 
            &= \mleft(\rho_{X_1 C_1}^{-1/4}\otimes\sigma_{C_2}^{-1/4} \mright) \mleft([x_1]_{X_1}\otimes\rho_{C_1 \land x_1}\otimes\rho_{C_2 \land x_2} \mright) \mleft(\rho_{X_1 C_1}^{-1/4}\otimes\sigma_{C_2}^{-1/4} \mright) \\
            &= \mleft(\rho_{X_1 C_1}^{-1/4} \mleft([x_1]_{X_1}\otimes\rho_{C_1 \land x_1}\mright) \rho_{X_1 C_1}^{-1/4} \mright) \otimes \mleft( \sigma_{C_2}^{-1/4} \rho_{C_2 \land x_2}\sigma_{C_2}^{-1/4} \mright) \\
            &= [x_1]_{X_1}\otimes \left(\rho_{C_1 \land x_1}\right)^{1/2} \otimes N_{C_2}(x_2).
        \end{split}
    \end{equation}
    In the last line we defined $N_{C_2}(x_2) \coloneqq \sigma_{C_2}^{-1/4} \rho_{C_2 \land x_2}\sigma_{C_2}^{-1/4}$ for compactness.
    With the equation above, the trace in~\cref{eq:applying_prop_proof_step_1_in_reaching_classical} becomes    
    \begin{equation}\label{eq:proof:reaching_classical:simplifying_trace}
        \begin{split}
            &\tr{\mleft(\sigma_{X_1 C}^{-1/4}\rho_{X_1 C \land x_1,x_2}\sigma_{X_1 C}^{-1/4}\mright) \mleft(\sigma_{X_1 C}^{-1/4}\rho_{X_1 C \land x_1',x_2'}\sigma_{X_1 C}^{-1/4}\mright)} \\
            &\qquad= \tr{  \mleft([x_1]_{X_1} [x_1']_{X_1}\mright) \otimes \mleft((\rho_{C_1 \land x_1})^{1/2} (\rho_{C_1 \land x_1'})^{1/2} \mright) \otimes \mleft(N_{C_2}(x_2) N_{C_2}(x_2') \mright)}\\
            &\qquad= |\langle x_1 |x_1' \rangle|^2\; \tr{ (\rho_{C_1 \land x_1})^{1/2} (\rho_{C_1 \land x_1'})^{1/2} } \; \tr{ N_{C_2}(x_2)N_{C_2}(x_2') } \\
            &\qquad= \delta_{x_1,x_1'} P_{X_1}(x_1) \tr{ N_{C_2}(x_2)N_{C_2}(x_2') },
        \end{split}
    \end{equation}
    where we used $\tr{ \rho_{C_1 \land x_1} }=P_{X_1}(x_1)$ in the last equality.
    
    Let us recall from~\cref{sec:Dodis_ext} that $s\cdot (\DEOR(x,y)) = (A_s x)\cdot y$ with the matrix $A_{s} = \sum_{i=1}^m A_i s_i$.
    By definition of the \textDEOR{}-extractor, $A_s$ has rank at least $n-r$.
    We define the new random variable $P_{(A_s \cdot X_1)}(x_1') = \sum_{x_1\;\mathrm{s.t.}\;x_1'=A_s x_1} P_{X_1}(x_1)$.
    By~\cref{prop:min_entropy_matrix_multipl} for trivial side information and using data-processing, we have $H_\mathrm{min}((A_s \cdot X_1)) \geq H_{\mathrm{min}}(X_1) - r \geq H_{\mathrm{min}}(X_1|C_1) - r \geq k_1 - r $. 
    From this we get the bound $P_{(A_s \cdot X_1)}(x_1')\leq 2^{-k_1 +r}$.
    We now return to~\cref{eq:applying_prop_proof_step_1_in_reaching_classical}. 
    Plugging in~\cref{eq:proof:reaching_classical:simplifying_trace}, we get
    \begin{equation}
        \begin{split}
            \delta&\bigl(\rho_{\DEOR(X_1,X_2) X_1 C}, \omega_{m}\otimes\rho_{X_1 C} \bigr)^2 \\
                  &\leq \frac{1}{4} \sum_{s\neq 0^m} \sum_{x_1,x_2,x_2'} (-1)^{(A_s x_1)\cdot x_2 + (A_s x_1)\cdot x_2'} P_{X_1}(x_1) \tr{N_{C_2}(x_2)N_{C_2}(x_2')}\\
                  &= \frac{1}{4} \sum_{s\neq 0^m} \sum_{x_1',x_2,x_2'} (-1)^{x_1'\cdot x_2 + x_1'\cdot x_2'} P_{(A_s\cdot X_1)}(x_1') \tr{N_{C_2}(x_2)N_{C_2}(x_2')}\\
                  &= \frac{1}{4} \sum_{s\neq 0^m} \sum_{x_1'} 2^n P_{(A_s\cdot X_1)}(x_1') \tr{\mathfrak{F}[N](x_1')\mathfrak{F}[N](x_1')}\\
                  &\leq 2^{-k_1 - 2 +n+r} \sum_{s\neq 0^m} \sum_{x_1'} \tr{\mathfrak{F}[N](x_1')^\dagger \mathfrak{F}[N](x_1')},
        \end{split}
    \end{equation}
    where we used linearity of the trace to identify the Fourier transform of $N$.
    The Fourier transform of $N$ is hermitian as each matrix $N(x_2)$ is hermitian. 
    Hence, each summand $\tr{\mathfrak{F}[N](x_1')\mathfrak{F}[N](x_1')}$ is positive, allowing us to apply our upper bound $P_{(A_s X_1)}(x_1') \leq 2^{-k_1+r}$ to obtain the last line.

    Note, that none of the summands in the equation above depend on $s$. Furthermore, we identify the L\textsubscript{2}-norm for matrix-valued functions and use~\cref{lem:Parseval} once again, to get
    \begin{equation*}
        \delta\bigl(\rho_{\DEOR(X_1,X_2) X_1 C} , \omega_{m}\otimes\rho_{X_1 C} \bigr)^2 
        \leq 2^{-k_1 - 2 +n+r + m} \trinorm{\mathfrak{F}[N]}_2^2
        = 2^{-k_1 - 2 +n+r+m} \trinorm{N}_2^2\\
    \end{equation*}
    By definition of the L\textsubscript{2}-norm and due to the form of the collision-entropy for cq-states (\cref{lem:collision_ent_cq}), we observe
    \begin{equation*}
        \trinorm{N}_2^2 =  \sum_{x_2} \tr{\mleft(\sigma_{C_2}^{-1/4} \rho_{C_2 \land x_2}\sigma_{C_2}^{-1/4}\mright)^2 } = 2^{-H_2(\rho_{X_2 C_2} | \sigma_{C_2})}.
    \end{equation*}
    We can choose $\sigma_{C_2}\in\mathcal{S}(\mathcal{H}_{C_2})$ such that the collision-entropy of $\rho_{X_2 C_2}$ relative to $\sigma_{C_2}$ is maximized and thus equal to $H_2(X_2|C_2)$. %\jakob{"We can choose $\sigma_{C_2}\in\mathcal{S}(\mathcal{H}_{C_2})$ such that..."?} 
    Observing that the collision-entropy upper bounds the min-entropy gives the final result
    \begin{equation*}
        \delta\bigl(\rho_{\DEOR(X_1,X_2) X_1 C} , \omega_{m}\otimes\rho_{X_1 C} \bigr)^2 
        \leq 2^{-k_1 -k_2- 2 +n+r+m}.
    \end{equation*}
\end{proof}

\subsection{Extension to the Markov Model}\label{sec:extension_to_Markov_model}
Up to this point we only presented results for two-source extractors against quantum product-type side information.
In this section, we extend our results to the case of quantum side information in the Markov model following the approach from \cite{arnonfriedman_extractors}.
In particular, the following lemma shows that a two-source extractor against quantum product-type side information remains secure against quantum side information in the Markov model with similar parameters. The lemma is analogous to \cite[lemma 6]{arnonfriedman_extractors} with an almost identical proof, although here we also allow for quantum side information.
\begin{lemma}\label{lem:quantum_Ext_to_Markov}
    Any $(k_1, k_2, \varepsilon)$ $X_1$-strong two-source extractor against quantum product-type knowledge is a $(k_1+\log\frac{1}{\varepsilon}, k_2 + \log\frac{1}{\varepsilon}, 3\varepsilon)$ $X_1$-strong two-source extractor against quantum knowledge in the Markov model.
\end{lemma}
\begin{proof}
    Let $\Ext$ be a $(k_1,k_2,\varepsilon)$ $X_1$-strong two-source extractor against quantum product-type knowledge and let $\rho_{X_1 X_2 C}$ be a Markov chain such that $H_{\mathrm{min}}(X_1|C) \geq k_1 + \log\frac{1}{\varepsilon}$ and $H_{\mathrm{min}}(X_2|C) \geq k_2 + \log\frac{1}{\varepsilon}$.
    \Cref{thm:Markov_rewriting} allows us to write the state as
    \begin{equation}
        \rho_{A_1 A_2 C} = \mathop{\bigoplus}_{z\in\mathcal{Z}}P_Z(z) \rho^z_{A_1 C_1^z} \otimes \rho^z_{A_2 C_2^z},
    \end{equation}
    which is isometric to
    \begin{equation}\label{eq:Markov_rewriting}
        \tilde{\rho}_{A_1 A_2 C_1 C_2 Z} = \sum_{z} P_Z(z) \tilde{\rho}_{A_1 C_1|z} \otimes \tilde{\rho}_{A_2 C_2|z} \otimes\ketbra{z}{z}_Z.
    \end{equation}
    through embedding the space $\bigoplus_z\hilmap_{C_1^z}\otimes\hilmap_{C_2^z}$ into a larger space.

\begin{comment}
    We recall from~\cref{eq:Markov_rewriting} that $\rho_{X_1 X_2 C}$ can be replaced by the equivalent state 
    \begin{equation}
        \tilde{\rho}_{X_1 X_2 C_1 C_2 Z} = \sum_z P_Z(z) \tilde{\rho}_{X_1 X_2 C_1 C_2|z} \otimes [z]_Z,
    \end{equation}
    where $\tilde{\rho}_{X_1 X_2 C_1 C_2|z} = \tilde{\rho}_{X_1 C_1 | z}\otimes\tilde{\rho}_{X_2 C_2 | z}$.
\end{comment}
    
    Then we have
    \begin{equation}
        2^{-H_{\mathrm{min}}(X_i|C)_\rho} 
        = 2^{-H_{\mathrm{min}}(X_i|C_1 C_2 Z)_{\tilde{\rho}}} 
        = \mathop{\mathbb{E}}_{z\leftarrow Z}\mleft[ 2^{-H_{\mathrm{min}}(X_i|C_1 C_2)_{\tilde{\rho}_{|z}}} \mright] 
        = \mathop{\mathbb{E}}_{z\leftarrow Z}\mleft[ 2^{-H_{\mathrm{min}}(X_i|C_i)_{\tilde{\rho}_{|z}}} \mright],
    \end{equation}
    where we used \cite[equation 6.25]{Tomamichel_2016} in the second equality.
    By Markov's inequality we get
    \begin{equation}
                \mathop{\mathrm{Pr}}_{z\leftarrow Z}\mleft[ H_{\mathrm{min}}(X_i|C_i)_{\tilde{\rho}_{|z}} < k_1 \mright] \leq \mathop{\mathrm{Pr}}_{z\leftarrow Z}\mleft[ H_{\mathrm{min}}(X_i|C_i)_{\tilde{\rho}_{|z}} + \log\frac{1}{\varepsilon} \leq H_{\mathrm{min}}(X_i|C)_\rho \mright]
        %= \mathop{\mathrm{Pr}}_{z\leftarrow Z}\mleft[ 2^{-H_{\mathrm{min}(X_i|C_i)_{\tilde{\rho}^z}}}\geq 2^{-H_{\mathrm{min}}(X_i|C)_\rho} /\varepsilon \mright] 
        \leq \varepsilon,
    \end{equation}
    which, by the union bound, implies that both $H_{\mathrm{min}}(X_1|C_1)_{\tilde{\rho}_{|z}} \geq k_1$ and $H_{\mathrm{min}}(X_2|C_2)_{\tilde{\rho}_{|z}} \geq k_2$ are fulfilled with probability at least $1-2\varepsilon$.
    We find
    \begin{equation}
    \begin{aligned}
        \delta\big(\rho_{\Ext(X_1,X_2)X_1 C}, \omega_{m}\otimes\rho_{X_1 C}\big) =& \delta\big(\tilde{\rho}_{\Ext(X_1,X_2)X_1 C_1 C_2 Z}, \omega_{m}\otimes\tilde{\rho}_{X_1 C_1 C_2 Z}\big) \\
        =& \mathop{\mathbb{E}}_{z\leftarrow Z}\mleft[ \delta\big(\tilde{\rho}_{\Ext(X_1,X_2)X_1 C_1 C_2|z}, \omega_{m}\otimes\tilde{\rho}_{X_1 C_1|z}\otimes\tilde{\rho}_{C_2|z}\big) \mright] \\
        \leq& 3 \varepsilon,
    \end{aligned}
    \end{equation}
    where we used~\cref{prop:trace_norm_cq} in line two.
    For the last line we bounded the $\delta\big(\tilde{\rho}_{\Ext(X_1,X_2) X_1C_1C_2|z}, \omega_{m}\otimes\tilde{\rho}_{X_1C_1|z}\otimes\tilde{\rho}_{C_2|z}\big)$ terms by 1 if $H_{\mathrm{min}}(X_1|C_1)_{\tilde{\rho}_{|z}} < k_1$ or $H_{\mathrm{min}}(X_2|C_2)_{\tilde{\rho}_{|z}} < k_2$, which happens with probability $2\varepsilon$ at most, and used the fact that $\Ext$ is a $X_1$-strong two-source extractor against product-type quantum knowledge else.
\end{proof}

Let us apply this result to our first main result,~\cref{thm:general_ext_result}, that we found for generic two-source extractors via our modular approach.
\begin{corollary}\label{cor:generic_ext_Markov_model}
    Any $(k_1, k_2, \varepsilon)$ $X_1$-strong two-source extractor, that produces an $m$-bit output, is a $(k_1+\frac{1}{2}\log\frac{1}{2^m\varepsilon}, k_2 + \frac{1}{2}\log\frac{1}{2^m\varepsilon}+\log\frac{1}{\varepsilon}, 3\sqrt{2^{m}\varepsilon})$ $X_1$-strong two-source extractor against quantum knowledge in the Markov model.
\end{corollary}
\begin{proof}
    Consider a $(k_1,k_2,\varepsilon)$ $X_1$-strong two-source extractor $\Ext:\{0,1\}^{n_1}\times\{0,1\}^{n_2}\rightarrow \{0,1\}^{m}$.
    We know from~\cref{cor:no_side_info_to_quantum_product} that $\Ext$ is a $(k_1',k_2',\varepsilon)$ $X_1$-strong two-source extractor against quantum side information in the product model, where $k_1'=k_1$, $k_2'=k_2+\log\frac{1}{\varepsilon}$ and $\varepsilon'=\sqrt{2^m\varepsilon}$.
    Then, according to~\cref{lem:quantum_Ext_to_Markov} $\Ext$ is a $(k_1'+\log\frac{1}{\varepsilon'}, k_2' + \log\frac{1}{\varepsilon'}, 3\varepsilon')$ $X_1$-strong two-source extractor against quantum knowledge in the Markov model.
    Plugging in the expressions for $k_1',k_2'$ and $\varepsilon'$ gives the desired statement. 
\end{proof}

Let us compare~\cref{cor:generic_ext_Markov_model} to \cite[lemma 9]{arnonfriedman_extractors} (also given in~\cref{prop:ext_to_ext_against_quantum_Markov}).
Note that, in both lemmas, additional terms have to be added to $k_1$ and $k_2$, in order to obtain a two-source extractor against quantum side information in the Markov model. 
In~\cref{prop:ext_to_ext_against_quantum_Markov} the sum of these additional terms is $2\log\frac{1}{\varepsilon}$, whereas in~\cref{cor:generic_ext_Markov_model} the sum of these additional terms is $\log\frac{1}{\varepsilon} + \log\frac{1}{2^m\varepsilon} = 2\log\frac{1}{\varepsilon}-m$.
Hence, the construction in~\cref{cor:generic_ext_Markov_model} is slightly stronger than the equivalent statement in \cite{arnonfriedman_extractors}.
%In the two lemmas $\varepsilon$ is, however, not referring to the same quantity. 
%On the one hand, $\varepsilon$ describes the error of the extractor when considering quantum product-type side information in Lemma \ref{lem:quantum_Ext_to_Markov}.
%For this reason, we denote it $\varepsilon_{q}$ in the following. 
%In \cite{arnonfriedman_extractors} on the other hand, $\varepsilon$ refers to the error of the extractor when no side information is considered. 
%Therefore, we refer to it as $\varepsilon_{no}$.
%Typically, $\varepsilon_q$ is exponentially larger than $\varepsilon_{no}$.
%Hence, $\log1/\varepsilon_q$ will be significantly smaller than $\log1/\varepsilon_{no}$.
%Correspondingly, the $\log1/\varepsilon_q$-term in Lemma \ref{lem:quantum_Ext_to_Markov} results in a much smaller reduction in the extractor's performance than the reduction caused by the $\log1/\varepsilon_{no}$-term in \cite{arnonfriedman_extractors}.
Finally, let us apply~\cref{lem:quantum_Ext_to_Markov} to our second main result, \cref{thm:BleA_1m_result}, that we obtained for the case of the \textDEOR{}-extractor via our non-modular approach.

\begin{corollary}\label{cor:BleA_against_Markov}
    The \textDEOR{}-extractor is a $(k_1,k_2,\varepsilon)$ $X_1$-strong two-source extractor against quantum knowledge in the Markov model, where
    \begin{equation*}
        \varepsilon = 3 \cdot 2^{-\frac{k_1+k_2+2-n-r-m}{4}}.
    \end{equation*}
\end{corollary}
\begin{proof}
    From~\cref{lem:quantum_Ext_to_Markov} and~\cref{thm:BleA_1m_result} we know that the \textDEOR{}-extractor is a $(k_1'+\log\frac{1}{\varepsilon'},k_2'+\log\frac{1}{\varepsilon'},3\varepsilon')$ $X_1$-strong two-source extractor against quantum knowledge in the Markov model, where $\varepsilon' = 2^{-\frac{k_1'+k_2'+2-n-r-m}{2}}$.
    Setting $k_1 = k_1' + \log\frac{1}{\varepsilon'}$,  $k_2 = k_2' + \log\frac{1}{\varepsilon'}$ and $\varepsilon=3\varepsilon'$ we find
    \begin{equation*}
        \varepsilon
        %= 2^{-\frac{k_1'+k_2'-n-r-m}{2}} 
        = 3 \cdot 2^{-\frac{k_1+k_2+2-n-r-m}{4}}
    \end{equation*}
    as desired.
\end{proof}
Comparing this result to~\cref{thm:BleA_1m_result} shows that, for $k_1, k_2, \varepsilon$ fixed, the possible output length $m$ in the Markov model decreases only logarithmically in $\varepsilon$ when compared to quantum side information of product form.
\section{Conclusion}
\label{chapter:conclusion}

In this work we have shown that the \textDEOR{}-extractor performs significantly better in the presence of both quantum product-type side information and quantum side information in the Markov model than previously known.
In particular, we provide two separate proofs that yield improvements over the results from Arnon-Friedman et al.~\cite{arnonfriedman_extractors} for the specific case of the \textDEOR{}-extractor.  

For the first approach, that we term modular, we introduce a new XOR-Lemma which we call the measured XOR-Lemma. 
It allows us to reduce the security of a two-source extractor against quantum side information to the security of a single-bit output two-source extractor against classical side information.
Here the classical side information is obtained from the quantum side information by application of a pretty good measurement.
\Cref{thm:general_ext_result} results from applying our modular approach to study arbitrary two-source exractors against quantum product-type side information.
In this case, the extractor's error displays the same $\sqrt{2^m \varepsilon}$ dependence that was already established in \cite{arnonfriedman_extractors}.
In addition, applying this modular approach to the \textDEOR{}-extractor allows us to relate the security of the \textDEOR{}-extractor to the security of the inner product extractor, yielding better results than directly applying the generic reductions from \cref{thm:general_ext_result} or \cite[corollary 16]{arnonfriedman_extractors}.

The second approach, which we call non-modular, is based on the Fourier transformation for matrix-valued functions. 
Against quantum product-type side information this approach gives \cref{thm:BleA_1m_result} which increases the \textDEOR{}-extractor's output length by a factor of 5 over \cite[corollary 21]{arnonfriedman_extractors}\footnote{Their results are stated in the Markov model which, of course, also apply to product-type side information.}.
Notably, \cref{thm:BleA_1m_result} implies that the \textDEOR{}-extractor displays the same parameters in the presence of quantum product-type side information as in the presence of no side information at all.
Furthermore, in~\cref{cor:BleA_against_Markov} we extend our results to the Markov model.
We remark that this result implies that the performance of the \textDEOR{}-extractor against side information in the Markov model is the same irrespective of whether the side information is quantum or classical. 

From our discussion above, two open questions arise. The first concerns arbitrary two-source extractors against quantum product-type side information and has already been posed in \cite{arnonfriedman_extractors,Berta_2017}. As discussed before, \cref{thm:general_ext_result} is, up to a constant factor, identical to the best known result \cite[corollary 16]{arnonfriedman_extractors}. This begs the question whether one can hope for future improvements when considering the security of arbitrary extractors against quantum side information. More precisely, it is unknown whether there exist specific examples of extractors which saturate the exponentially growing error term $\sqrt{2^m \varepsilon}$ in~\cref{thm:general_ext_result} (as noted in \cite{Berta_2017} the fastest growing error known so far is $m \varepsilon$ \cite{Gavinsky_2007}). Finding such examples, or showing that they don't exist, presents an interesting direction for future research.

The second direction of research is related to our measured XOR-Lemma and the performance improvement it yields for the \textDEOR{}-extractor. 
Due to the modularity of our proof-technique, it can easily be adapted to study the security of other extractors besides the \textDEOR{}-extractor.
We are optimistic that this fact will improve the performance of a large class of extractors.
Because of the nature of the measured XOR-Lemma, extractors $\Ext$ that give good single-bit output extractors $s \cdot \Ext$ are of particular interest.
Finding explicit constructions of extractors belonging to this class promises to be a fruitful topic of future research. One example of particular interest is the extractor from \cite{Raz_2005} which can extract randomness even when one of the two sources has only logarithmic entropy.

\bigskip

\paragraph{Acknowledgments.}
We thank Renato Renner, Ramona Wolf, and Cameron Foreman for valuable comments and discussion. We thank Ramona Wolf and Renato Renner for help in supervising the project. This
work was supported by the NSF project No.~20QU-1\_225171 and the CHIST-ERA project ``Modern Device Independent Cryptography"
MS and CF acknowledge support from NCCR SwissMAP and the ETH Zurich Quantum Center.

\begin{comment}
    \carla{Cameron, add people...} \martin{Add funding}
    I want to thank Martin Sandfuchs, Carla Ferradini, Prof. Dr. Ramona Wolf and Prof. Dr. Renato Renner for their great ideas, time and effort that they have put into this manuscript.
    The many discussions and meetings with them were essential to the project's success.
    Additionally I want to thank them for inviting me to visit the quantum key distribution summer school in Les Diablerets and the amazing time that I spent there.
\end{comment}

\bibliographystyle{halpha}
%\bibliography{references}

\newcommand{\etalchar}[1]{$^{#1}$}

\appendix
\section{Proofs for two-source extractors against classical side information}
\label{app:A}
\subsection{Proof of proposition \ref{prop:ext_to_ext_against_classical_product_type}}
\label{app:proof:ext_to_ext_against_classical_product_type}

\extagainstcpt*

\begin{proof}
    Let $\Ext$ be a $(k_1,k_2,\varepsilon)$ $X_1$-strong two-source extractor.
    Let $X_1$ and $X_2$ be two sources of randomness with classical side information of the form $\rho_{X_1 Z_1}\otimes \rho_{X_2 Z_2}$ such that $H_{\min}(X_1|Z_1)\geq k_1$ and $H_{\min}(X_2|Z_2)\geq k_2 + \log\frac{1}{\varepsilon}$. 
    By definition of the conditional min-entropy we find that 
    \begin{equation*}
        2^{-H_{\min}(X_2|Z_2)} = \sum_{z_2} P_{Z_2}(z_2) \max_{x_2} P_{X_2|Z_2=z_2}(x_2) = \mathop{\mathbb{E}}_{Z_2 \leftarrow z_2}\mleft[2^{-H_{\min}(X_2|{Z_2=z_2})}\mright],
    \end{equation*}
    where $H_{\min}(X_2|{Z_2=z_2})$ stands for the min-entropy of the probability distribution $P_{X_2|Z_2=z_2}(x_2)$.
    We find
    \begin{align*}
        \begin{split}
            \mathop{\mathrm{Pr}}_{z_2 \leftarrow Z_2}\mleft[ H_{\min}(X_2|{Z_2=z_2}) < k_2 \mright] 
            &\leq \mathop{\mathrm{Pr}}_{z_2\leftarrow Z_2}\mleft[ H_{\min}(X_2|{Z_2=z_2})\leq H_{min}(X_2|Z_2) - \log\frac{1}{\varepsilon} \mright] \\
            &= \mathop{\mathrm{Pr}}_{z_2\leftarrow Z_2}\mleft[ 2^{-H_{\min}(X_2|{Z_2=z_2})}\geq \frac{1}{\varepsilon} 2^{-H_{min}(X_2|Z_2)}\mright] \\
            &\leq \mathop{\mathbb{E}}_{z_2 \leftarrow Z_2}\mleft[2^{-H_{\min}(X_2|{Z_2=z_2})}\mright]  2^{H_{min}(X_2|Z_2)} \varepsilon \\
            &= \varepsilon,
        \end{split}
    \end{align*}
    where we used Markov's inequality in the last inequality.
    This means that the probability distribution $P_{X_2|Z_2=z_2}(x_2)$ can be interpreted as a source of randomness that has min-entropy larger than $k_2$ with probability greater than $1-\varepsilon$.
    Note, that $H_{\min}(X_1)\geq H_{\min}(X_1|Z_1)\geq k_1$ due to data-processing of the min-entropy as given in~\cref{lem:data_processing}. Let us now consider
    \begin{equation*}
    \begin{aligned}
        \delta\big(\rho_{\Ext(X_1,X_2)X_1 Z_1 Z_2}, \omega_{m}\otimes\rho_{X_1 Z_1 Z_2}\big) 
        &= \delta\big(\rho_{\Ext(x_1,X_2) X_1 Z_2}, \omega_{m}\otimes\rho_{X_1 Z_2}\big) \\
        &= \mathop{\mathbb{E}}_{z_2\leftarrow Z_2}\mleft[ \delta\big(\rho_{\Ext(X_1,X_2)X_1|z_2}, \omega_{m}\otimes\rho_{X_1}\big) \mright]\\
        & \leq 2\varepsilon,
    \end{aligned}
    \end{equation*}
    where we used~\cref{eq:strong_ext_condition} and that the sources $X_1$, $X_2|_{Z_2=z_2}$ have min-entropy of at least $k_1$, $k_2$ with probability greater than $1-\varepsilon$. For the last line we used that $\Ext$ is a $(k_1,k_2,\varepsilon)$ $X_1$-strong two-source extractor.
\end{proof}

\section{Weak extractors} \label{sec:weak_ext_proof}

\begin{corollary}\label{cor:measured_XOR_weak_ext}
    Let $\rho_{X_1 X_2 C_1 C_2} = \rho_{X_1 C_1}\otimes\rho_{X_2 C_2}$, where $\rho_{X_1 C_1}$ and $\rho_{X_2 C_2}$ are cq-states. For any function $\Ext: \{0,1\}^{n_1}\times\{0,1\}^{n_2}\mapsto\{0,1\}^m$, we have
    \begin{equation*}
        \begin{split}
            \delta\big(&\rho_{\Ext(X_1,X_2) C_1 C_2}, \omega_{m} \otimes \rho_{C_1} \otimes \rho_{C_2}\big)\\
            &\leq \sqrt{ \frac{1}{2} \sum_{0\neq s \in \{0,1\}^m} \delta\left( \PGMchannelMartin{\rho_{X_1 C_1}}{X'_1|C_1} \otimes \PGMchannelMartin{\rho_{X_2 C_2}}{X'_{2}|C_2}[\rho_{(s \cdot \Ext(X_1,X_2)) C_1 C_2}], \omega_{1} \otimes \PGMchannelMartin{\rho_{X_1 C_1}}{X_1'|C_1}[\rho_{C_1}] \otimes \PGMchannelMartin{\rho_{X_2 C_2}}{X'_2|C_2}[\rho_{C_2}] \right) }. \\
        \end{split}
    \end{equation*}
\end{corollary}

\begin{proof}
We begin by applying the measured XOR-Lemma (\cref{lem:measured_XOR_lemma}) with $Z = \mathrm{Ext}(X_1, X_2)$ and $E = C_1 C_2$ which yields
\begin{equation*}
\begin{aligned}
    \delta\big(&\rho_{\Ext(X_1,X_2) C_1 C_2}, \omega_{m} \otimes \rho_{C_1} \otimes \rho_{C_2}\big)^2 \\
    &\leq \frac{1}{2} \sum_{s \neq 0} \delta \left( \PGMchannelMartin{\rho_{(s \cdot \mathrm{Ext}(X_1, X_2)) C_1 C_2}}{I|C_1 C_2}[\rho_{(s \cdot \mathrm{Ext}(X_1, X_2)) C_1 C_2}] , \omega_{1} \otimes  \PGMchannelMartin{\rho_{(s \cdot \mathrm{Ext}(X_1, X_2)) C_1 C_2}}{I|C_1 C_2}[\rho_{C_1 C_2}] \right).
\end{aligned}
\end{equation*}
For ease of notation, let us define $f^{s}(x_1, x_2) \coloneqq s \cdot \mathrm{Ext}(x_1, x_2)$. With this, we get
\begin{equation}
\begin{aligned} \label{eq:weak_ext_1}
    \delta\big(&\rho_{\Ext(X_1,X_2) C_1 C_2}, \omega_{m} \otimes \rho_{C_1} \otimes \rho_{C_2}\big)^2 \\
    &\leq \frac{1}{2} \sum_{s \neq 0} \delta \left( \PGMchannelMartin{\rho_{f^{s}(X_1, X_2) C_1 C_2}}{I|C_1 C_2}[\rho_{(s \cdot \mathrm{Ext}(X_1, X_2)) C_1 C_2}] , \omega_{1} \otimes  \PGMchannelMartin{\rho_{f^{s}(X_1, X_2) C_1 C_2}}{I|C_1 C_2}[\rho_{C_1 C_2}] \right) \\
    &= \frac{1}{2} \sum_{s \neq 0} \delta \left( f^{s}_{I|X'_1 X'_2} \circ \PGMchannelMartin{\rho_{X_1 X_2 C_1 C_2}}{X'_1 X'_2|C_1 C_2}[\rho_{(s \cdot \mathrm{Ext}(X_1, X_2)) C_1 C_2}] , \omega_{1} \otimes  f^{s}_{I|X'_1 X'_2} \circ \PGMchannelMartin{\rho_{X_1 X_2 C_1 C_2}}{X'_1 X'_2|C_1 C_2}[\rho_{C_1 C_2}] \right),
\end{aligned}
\end{equation}
where we applied~\cref{prop:PGM_f_commute} for the equality. Using the data-processing inequality given in ~\cref{lem:trace_norm_data_processing}, we obtain
\begin{equation}
\begin{aligned} \label{eq:weak_ext_2}
    & \delta \left( f^{s}_{I|X'_1 X'_2} \circ \PGMchannelMartin{\rho_{X_1 X_2 C_1 C_2}}{X'_1 X'_2|C_1 C_2}[\rho_{(s \cdot \mathrm{Ext}(X_1, X_2)) C_1 C_2}] , \omega_{1} \otimes  f^{s}_{I|X'_1 X'_2} \circ \PGMchannelMartin{\rho_{X_1 X_2 C_1 C_2}}{X'_1 X'_2|C_1 C_2}[\rho_{C_1 C_2}] \right) \\
    &\quad \leq \delta \left( \PGMchannelMartin{\rho_{X_1 X_2 C_1 C_2}}{X'_1 X'_2|C_1 C_2}[\rho_{(s \cdot \mathrm{Ext}(X_1, X_2)) C_1 C_2}] , \omega_{1} \otimes \PGMchannelMartin{\rho_{X_1 X_2 C_1 C_2}}{X'_1 X'_2|C_1 C_2}[\rho_{C_1 C_2}] \right).
\end{aligned}
\end{equation}
From~\cref{def:PGM}, one can easily see that since $\rho_{X_1 X_2 C_1 C_2}$ is a product state, the PGM factorizes:
\begin{equation} \label{eq:pgm_factorizes}
    \PGMchannelMartin{\rho_{X_1 X_2 C_1 C_2}}{X'_1 X'_2|C_1 C_2} = \PGMchannelMartin{\rho_{X_1 C_1}}{X'_1|C_1} \otimes \PGMchannelMartin{\rho_{X_2 C_2}}{X'_2|C_2}.
\end{equation}
Combining~\cref{eq:pgm_factorizes,eq:weak_ext_1,eq:weak_ext_2} then yields the desired result.
\end{proof}

\begin{lemma}
    Any $(k_1, k_2, \varepsilon)$ weak two-source extractor against classical product-type knowledge is a $(k_1, k_2,\allowbreak \sqrt{2^{m-1}\varepsilon})$ weak two-source extractor against quantum product-type knowledge.
\end{lemma}

\begin{proof}
Let $\rho_{X_1 X_2 C_1 C_2} = \rho_{X_1 C_1} \otimes \rho_{X_2 C_2}$ be a state such that $H_{\min}(X_1|C_1) \geq k_1$ and $H_{\min}(X_2|C_2) \geq k_2$. From~\cref{cor:measured_XOR_weak_ext} we know that
\begin{equation} \label{eq:weak_ext_3}
    \begin{split}
        \delta\big(&\rho_{\Ext(X_1,X_2) C_1 C_2}, \omega_{m} \otimes \rho_{C_1} \otimes \rho_{C_2}\big)\\
        &\leq \sqrt{ \frac{1}{2} \sum_{0\neq s \in \{0,1\}^m} \delta\left( \PGMchannelMartin{\rho_{X_1 C_1}}{X'_1|C_1} \otimes \PGMchannelMartin{\rho_{X_2 C_2}}{X'_{2}|C_2}[\rho_{(s \cdot \Ext(X_1,X_2)) C_1 C_2}], \omega_{1} \otimes \PGMchannelMartin{\rho_{X_1 C_1}}{X_1'|C_1}[\rho_{C_1}] \otimes \PGMchannelMartin{\rho_{X_2 C_2}}{X'_2|C_2}[\rho_{C_1 C_2}] \right) } \\
        &\leq \sqrt{ \frac{1}{2} \sum_{0\neq s \in \{0,1\}^m} \delta\left( \PGMchannelMartin{\rho_{X_1 C_1}}{X'_1|C_1} \otimes \PGMchannelMartin{\rho_{X_2 C_2}}{X'_{2}|C_2}[\rho_{\Ext(X_1,X_2) C_1 C_2}], \omega_{m} \otimes \PGMchannelMartin{\rho_{X_1 C_1}}{X_1'|C_1}[\rho_{C_1}] \otimes \PGMchannelMartin{\rho_{X_2 C_2}}{X'_2|C_2}[\rho_{C_1} \otimes \rho_{C_2}] \right) }, \\
    \end{split}
\end{equation}
where for the second inequality we applied the data-processing inequality (\cref{lem:trace_norm_data_processing}).
From the data-processing inequality for the min-entropy (\cref{lem:data_processing}), we have that $H_{\min}(X_1|X'_1) \geq k_1$ and $H_{\min}(X_2|X'_2) \geq k_2$. Since $\mathrm{Ext}$ is a $(k_1, k_2, \varepsilon)$-extractor against classical product-type knowledge, we get
\begin{equation*}
    \delta\left( \PGMchannelMartin{\rho_{X_1 C_1}}{X'_1|C_1} \otimes \PGMchannelMartin{\rho_{X_2 C_2}}{X'_{2}|C_2}[\rho_{\Ext(X_1,X_2) C_1 C_2}], \omega_{m} \otimes \PGMchannelMartin{\rho_{X_1 C_1}}{X_1'|C_1}[\rho_{C_1}] \otimes \PGMchannelMartin{\rho_{X_2 C_2}}{X'_2|C_2}[\rho_{C_2}] \right) \leq \varepsilon. \\
\end{equation*}
Combining this with~\cref{eq:weak_ext_3} then gives the claimed result.
\end{proof}

\section{Proof of lemma \ref{lem:BleA_2m_result}}\label{app:B}
In this chapter we expand on the observation, stated in~\cref{sec:proof_theorem_1}, that the measured XOR-Lemma allows for simple improvements in the extractor's output length if one considers specific extractor constructions instead of generic two-source extractors.
We demonstrate this for the case of the \textDEOR{}-extractor which results in~\cref{lem:BleA_2m_result}.
The idea for proving~\cref{lem:BleA_2m_result} is to show that $(s \cdot \DEOR)$ has a much smaller error than the extractor $\DEOR$.
To prove this, we first show that the inner product is a good two-source extractor (\cref{app:inner_product_extractor}) and then notice that $(s \cdot \DEOR)$ can be seen as applying the inner product extractor to the source $A_{s}^{T} \cdot X_{2}$ (\cref{app:proof:lem:BleA_2m_result}).

\subsection{Two-universal hashing}\label{app:two_universal_hashing}
Hash functions are frequently used in computer science applications such as password storage or message authentication to name just a few \cite{modern_crypto}.
Furthermore, it is well-known that two-universal families of hash functions constitute strong two-source extractors \cite{Renner_PhD_thesis, Hayashi2}. 

\begin{definition}[Two-universality]
    Let $\mathcal{F}$ be a family of functions from $\{0,1\}^{n_2}$ to $\{0,1\}^m$ that is indexed by $n_1$-bit strings, i.e.\ $\mathcal{F}=\{ f_{x_1}: \{0,1\}^{n_2} \rightarrow \{0,1\}^m \}_{x_1\in\{0,1\}^{n_1}}$. 
    We call $\mathcal{F}$ a two-universal family of hash functions, if
    \begin{equation}\label{eq:two_universality}
        \mathop{Pr}_{x_1\leftarrow U_{n_1}}\mleft[f_{x_1}(x_2) = f_{x_1}(x_2')\mright] \leq 2^{-m},
    \end{equation}
    for any two $n_2$-bit strings $x_2\neq x_2'$.
\end{definition}
Recall from~\cref{sec:QM_notation}, that we can view any classical function as a channel acting on a quantum system of suitable dimension.
Thus, we can interpret a two-universal family of hash functions $\mathcal{F}$ as a channel, where the random variable $X_1$ determines which element of $\mathcal{F}$ is applied.
If this choice is made uniformly at random, we find the following lemma to upper bound the L\textsubscript{2}-distance from uniform.
\begin{lemma}[{\cite[Lemma 5.4.3]{Renner_PhD_thesis}}]\label{lem:Renner_hashing}
    Let $\rho_{X_2 B}\in\mathcal{P}(\mathcal{H}_{X_2}\otimes\mathcal{H}_B)$ be a (non-normalized) cq-state that is classical on $X_{2}$ and let $\sigma_B\in\mathcal{P}(\mathcal{H}_B)$. 
    Let $\mathcal{F}=\{f_{x_1}:\{0,1\}^{n_2}\rightarrow\{0,1\}^m\}_{x_1\in\{0,1\}^{n_1}}$ be a two-universal family of hash functions.
    Then we have
    \begin{equation*}
        \mathop{\mathbb{E}}_{x_1\leftarrow U_{n_1}}\mleft[ d_2(\rho_{f_{x_1}(X_2)B}|\sigma_B) \mright] \leq \tr{\rho_{X_2 B}} 2^{-H_2(\rho_{X_2 B}|\sigma_B)}.
    \end{equation*}
\end{lemma}
As we are interested in cases where $X_1$ is not uniform, the following result allows us to reduce this case to the case of uniform $X_1$.
\begin{lemma}[{\cite[Theorem 7]{Hayashi2}}]\label{lem:Hayashi_Thm_7}
    Let $\rho_{X_2 B}\in\mathcal{P}(\mathcal{H}_{X_2}\otimes\mathcal{H}_B)$ be a (non-normalized) cq-state that is classical on $X_2$ and let $\sigma_B\in\mathcal{P}(\mathcal{H}_B)$. 
    Let $\mathcal{F}=\{f_{x_1}:\{0,1\}^{n_2}\rightarrow\{0,1\}^m\}_{x_1\in\{0,1\}^{n_1}}$ be a family of hash functions and let $X_1$ be a random variable over the alphabet of all $n_1$-bit strings, satisfying $H_{\min}(X_1) \geq k_1$.
    Then, we have
    \begin{equation*}
        \mathop{\mathbb{E}}_{x_1 \leftarrow X_1} \mleft[d_2(\rho_{f_{x_1}(X_2)B}|\sigma_B)\mright] \leq 2^{(n_1-k_1)}  \mathop{\mathbb{E}}_{x_1 \leftarrow U_{n_1}} \mleft[ d_2(\rho_{f_{x_1} (X_2)B}|\sigma_B) \mright].
    \end{equation*}
\end{lemma}
Putting the two lemmas above together allows us to rederive the following result.
\begin{lemma}[{\cite[Equation 115]{Hayashi2}}]\label{lem:Hayashi_hashing}
    Let $\mathcal{F}=\{f_{x_1}:\{0,1\}^{n_2}\rightarrow\{0,1\}^m\}_{x_1\in\{0,1\}^{n_1}}$ be a two-universal family of hash functions.
    Let $X_1$ be a random variable over the alphabet of all $n_1$-bit strings, satisfying $H_{\min}(X_1) \geq k_1$.
    Let $\rho_{X_2 Z_2}\in\mathcal{S}(\mathcal{H}_{X_2}\otimes\mathcal{H}_{Z_2})$ be a state that is classical on both $X_2$ and $Z_2$, such that $H_{\min}(X_2|Z_2)\geq k_2$.
    Then, we have
    \begin{equation}
        \mathop{\mathbb{E}}_{x_1\leftarrow X_1} \mleft[\delta\bigl(\rho_{f_{x_1}(X_2)Z_2} , \omega_{m}\otimes \rho_{Z_2} \bigr)\mright] \leq 2^{-\frac{1}{2}(2+k_1+k_2-n_1-m)}.
    \end{equation}
\end{lemma}
\begin{proof}
    Let $\sigma_{Z_2}\in\mathcal{S}(\mathcal{H}_{Z_2})$ arbitrary.
    Consider
    \begin{align*}
        \mathop{\mathbb{E}}_{x_1\leftarrow X_1} \mleft[\delta\bigl(\rho_{f_{x_1}(X_2)Z_2} , \rho_{U_m}\otimes \rho_{Z_2} \bigr)\mright]
        &\leq \frac{1}{2} \mathop{\mathbb{E}}_{x_1\leftarrow X_1} \mleft[\sqrt{2^m d_2(\rho_{f_{x_1}(X_2)Z_2}|\sigma_{Z_2})}\mright] \\
        &\leq \frac{1}{2} \sqrt{2^m \mathop{\mathbb{E}}_{x_1\leftarrow X_1} \mleft[ d_2(\rho_{f_{x_1}(X_2)Z_2}|\sigma_{Z_2})\mright]}\\
        &\leq \frac{1}{2} \sqrt{2^{m + n_1-k_1} \mathop{\mathbb{E}}_{x_1 \leftarrow U_{n_1}} \mleft[ d_2(\rho_{f_{x_1}(X_2)Z_2}|\sigma_{Z_2})\mright]}\\
        &\leq \frac{1}{2} \sqrt{2^{m + n_1-k_1 -H_2(\rho_{X_2 Z_2}|\sigma_{Z_2})}}\\
        &\leq 2^{-(2+k_1 + H_{\min}(\rho_{X_2 Z_2}|\sigma_{Z_2}) - n_1 - m)/2}
    \end{align*}
    where we applied~\cref{cor:bound_one_norm_by_two_norm} to obtain the first inequality.
    In the second inequality, we used Jensen's inequality.
    For the third and fourth inequality we used~\cref{lem:Hayashi_Thm_7} and~\cref{lem:Renner_hashing}, respectively, together with the two-universality of $\mathcal{F}$.
    The last inequality can be obtained by using the fact the collision entropy upper bounds the min-entropy, as seen in \cref{sec:entropies}.
    Since $\sigma_{Z_2}$ was arbitrary, we can choose it such that the min-entropy of $\rho_{X_2 Z_2}$ relative to $\sigma_{Z_2}$ is maximized and thus equal to $H_{\min}(X_2|Z_2)$.
    By assumption, we have that $H_{\min}(X_2|Z_2) \geq k_2$ which concludes the proof.
\end{proof}

\subsection{Inner product single-bit output extractor}\label{app:inner_product_extractor}
One of the simplest two-source extractors that produces a one-bit output is given by the inner product modulo 2 denoted by
\begin{align*}
    \mathrm{IP}: \{0,1\}^n \times \{0,1\}^n &\longrightarrow \{0,1\} \\
    (x,y) &\longmapsto x^T y = \left(\sum_{i=1}^n x_i y_i\right) \mod{2},
\end{align*}
which we often abbreviate as $\mathrm{IP}(x,y)=x\cdot y$.
We also define the bitwise addition modulo 2, or bitwise XOR, of two bit strings $x,y\in\{0,1\}^n$ by $x\oplus y = (x_1 + y_1\mod{2},...,x_n + y_n \mod{2})$.

By showing that the inner product function $\mathrm{IP}$, gives a two-universal family of hash functions, we can prove that $\mathrm{IP}$ is a good two-source extractor.
\begin{lemma}[Two-universality of the inner product]
    The set $\{\mathrm{IP}(y,\cdot): \{0,1\}^n\rightarrow \{0,1\}\}_{y\in\{0,1\}^n}$ is a two-universal family of hash functions.
\end{lemma}
\begin{proof}
    Let $x,x'\in\{0,1\}^n$ such that $x\neq x'$. Then
    \begin{align*}
        \mathop{Pr}_{y\leftarrow U_n}\mleft[y\cdot x = y\cdot x'\mright] 
        &= \sum_{y\in\{0,1\}^n} 2^{-n} \frac{1+(-1)^{y\cdot x + y\cdot x'}}{2}\\
        &= \frac{1}{2} + 2^{-n-1}\sum_{y\in\{0,1\}^n}(-1)^{y\cdot(x\oplus x')}\\
        &= \frac{1}{2} + 2^{-n-1}\prod_{i=1}^n\mleft(\sum_{y_i\in\{0,1\}} (-1)^{y_i(x\oplus x')_i}\mright)\\
        &= \frac{1}{2} + 2^{-n-1}\prod_{i=1}^n\mleft(1+(-1)^{(x\oplus x')_i}\mright)\\
        &= \frac{1}{2}
    \end{align*}
    where we used in the last equality that for $x\neq x'$ there has to exist at least one $i$ such that $(x\oplus x')_i$ is nonzero.
\end{proof}
Together with \cref{lem:Hayashi_hashing} we immediately find that the inner product is a good extractor as formalized in the following corollary.
\begin{corollary}\label{cor:IP_strong_cl_ext}
    The inner product $\mathrm{IP}$ is a $(k_1,k_2,\epsilon)$ $X_1$-strong two-source extractor against classical product-type knowledge, where
    \begin{equation*}
        \epsilon = 2^{-\frac{1}{2}(1+k_1+k_2-n)}.
    \end{equation*}
\end{corollary}
\begin{proof}
    Let  $X_1$, $X_2$ be randomness sources with classical side information $Z$ of the form $\rho_{X_1 X_2 Z}=\rho_{X_1 Z_1}\otimes\rho_{X_2 Z_2}$ that satisfies $H_{\min}(X_1|Z_1)\geq k_1$ and  $H_{\min}(X_2|Z_2)\geq k_2$.
    Then by data-processing $H_{\min}(X_1)\geq H_{\min}(X_1|Z_1)\geq k_1$.
    We first consider the condition for the extractor to be $X_1$-strong, for which we observe
    \begin{equation*}
        \delta\bigl( \rho_{\mathrm{IP}(X_1,X_2)X_1 Z}, \rho_{U_1}\otimes\rho_{X_1 Z_1}\otimes\rho_{Z_2}\bigr) = \mathop{\mathbb{E}}_{x_1\leftarrow X_1} \mleft[\delta\big(\rho_{\mathrm{IP}(x_1,X_2)Z_2}, \rho_{U_1}\otimes\rho_{Z_2}\big)\mright]
        \leq 2^{-\frac{1}{2}(1+k_1+k_2-n)}
    \end{equation*}
    where we used the rewriting of the strong extractor condition (\cref{eq:strong_ext_condition}) in the equality.
    In order to obtain the inequality, we applied~\cref{lem:Hayashi_hashing} together with the fact that the set $\{\mathrm{IP}(x_1,\cdot)\}_{x_1\in\{0,1\}^n}$ is two-universal.
\end{proof}

\subsection{Proof of lemma~\ref{lem:BleA_2m_result}}\label{app:proof:lem:BleA_2m_result}
In the previous section we proved that the inner product is a good two-source extractor which essentially shows that the function $(s\cdot\DEOR)$ is a good two-source extractor as well. Guided by this idea we find the following Lemma:

\bleAres*

\begin{proof}
    Let $X_1$ and $X_2$ be randomness sources with quantum side information $C_1 C_2$ of product form $\rho_{X_1 X_2 C_1 C_2} = \rho_{X_1 C_1}\otimes \rho_{X_2 C_2}$ such that $H_{\min}(X_1|C_1)\geq k_1$ and $H_{\min}(X_2|C_2)\geq k_2$.
    Using~\cref{cor:measured_XOR_strong_ext} we have
    \begin{equation}\label{eq:proof:BleA_2m_result:applying_XOR_lemma}
    \begin{split}
        \delta\big(&\rho_{\DEOR(X_1,X_2)X_1 C}, \omega_{m}\otimes\rho_{X_1 C_1}\otimes\rho_{C_2}\big) \\
        &\leq \sqrt{ \frac{1}{2} \sum_{s\neq0^m} \delta \left( \PGMchannelMartin{\rho_{X_{2}C_{2}}}{X'_{2}|C_{2}} \left[ \rho_{(s\cdot \DEOR(X_1,X_2)) X_1 C_{2}} \right] , \omega_{1} \otimes \rho_{X_1} \otimes \PGMchannelMartin{\rho_{X_{2}C_{2}}}{X'_{2}|C_{2}} \left[ \rho_{C_2} \right] \right) } \\
        &= \sqrt{ \frac{1}{2} \sum_{s\neq0^m} \delta\left( \PGMchannelMartin{\rho_{X_{2}C_{2}}}{X'_{2}|C_{2}} \left[ \rho_{\mathrm{IP}(X_1,(A_s^T\cdot X_2)) X_1 C_2} \right] , \omega_{1} \otimes \rho_{X_1} \otimes \PGMchannelMartin{\rho_{X_2C_2}}{X'_2|C_2} \left[ \rho_{C_2} \right] \right) },
    \end{split}
    \end{equation}
    where we observed that $s\cdot\DEOR(x_1,x_2) =  x_1^T A_s^T x_2 = \mathrm{IP}(x_1, A_s^T x_2)$ in the last line.
    By assumption the matrix $A_s$, and thereby also the matrix $A_s^T$, has rank at least $n-r$ for any $0\neq s\in\{0,1\}^m$.
    From~\cref{prop:min_entropy_matrix_multipl} we then get 
    \begin{equation}
        H_{\min}((A_s^T \cdot X_2)|X_2') \geq H_{\min}(X_2|X_2') - r \geq H_{\min}(X_2|C_2) - r \geq k_2 - r
    \end{equation}
    where we used data-processing in the second to last inequality. Similarly we have that $H_{\min}(X_1) \geq k_1$.
    
    Recall from~\cref{cor:IP_strong_cl_ext} that $\mathrm{IP}$ is an $X_1$-strong extractor which immediately gives
    \begin{equation}
        \delta\left( \PGMchannelMartin{\rho_{X_{2}C_{2}}}{X'_{2}|C_{2}} \left[ \rho_{\mathrm{IP}(X_1,(A_s^T\cdot X_2)) X_1 C_2} \right] , \omega_{1} \otimes \rho_{X_1} \otimes \PGMchannelMartin{\rho_{X_2C_2}}{X'_2|C_2} \left[ \rho_{C_2} \right] \right)
        \leq 2^{-\frac{1}{2}(1+k_1+k_2-r-n)}.
    \end{equation}
    By plugging this back into~\cref{eq:proof:BleA_2m_result:applying_XOR_lemma} we obtain
    \begin{equation}
        \delta\big(\rho_{\DEOR(X_1,X_2)X_1 C}, \omega_{m}\otimes\rho_{X_1 C_1}\otimes\rho_{C_2}\big) \leq 2^{-\frac{k_1+k_2+3-r-n-2m}{4}}
    \end{equation}
    as desired.
\end{proof}

% \printbibliography

\end{document}